\documentclass[letterpaper,journal]{IEEEtran}
\IEEEoverridecommandlockouts
\usepackage{amsmath,amssymb,amsfonts}
\usepackage{amsthm}
\usepackage{algorithm}
\usepackage{algpseudocode}
\newtheorem{myDef}{Definition}
\newtheorem{myTheo}{Theorem}
\newtheorem{myLemm}{Lemma}

\usepackage{cite}
\usepackage{makecell}
\usepackage{booktabs}
\usepackage{multirow}
\usepackage{graphicx}
\usepackage{float} 
\usepackage{subfigure}
\usepackage{pifont}
\usepackage{textcomp}
\usepackage{color}
\newtheorem*{theorem*}{Theorem}
\usepackage{bm}
\usepackage{threeparttable}
\begin{document}

\title{\huge A QoE-Driven Personalized Incentive Mechanism Design for AIGC Services in Resource-Constrained Edge Networks}

\author{Hongjia~Wu, Minrui~Xu,
	Zehui~Xiong, Lin Gao,~\IEEEmembership{Senior Member,~IEEE}, Haoyuan Pan,~\IEEEmembership{Member,~IEEE},  Dusit~Niyato,~\IEEEmembership{Fellow,~IEEE}, and Tse-Tin Chan,~\IEEEmembership{Member,~IEEE}
\thanks{\emph{(Corresponding author: Tse-Tin Chan.)}}
\thanks{H.~Wu and T.-T.~Chan are with the Department of Mathematics and Information Technology, The Education University of Hong Kong, Hong Kong SAR, China (e-mails: {whongjia@eduhk.hk}, {tsetinchan@eduhk.hk}).}
\thanks{Z. Xiong is with the School of Electronics, Electrical Engineering and Computer Science (EEECS), Queen's University Belfast, Belfast, BT7 1NN, U.K. (e-mail: z.xiong@qub.ac.uk).}
\thanks{L. Gao is with the School of Electronics and Information Engineering and the Guangdong Provincial Key Laboratory of Aerospace Communication and Networking Technology, Harbin Institute of Technology, Shenzhen, China (e-mail: gaol@hit.edu.cn).}
\thanks{H. Pan is with the College of Computer Science and Software Engineering, Shenzhen University, Shenzhen, China (e-mail: hypan@szu.edu.cn).}
\thanks{M. Xu and D. Niyato are with the College of Computing and Data Science, Nanyang Technological University, Singapore (e-mails: {minrui001@e.ntu.edu.sg}, {dniyato@ntu.edu.sg}).}
}

\maketitle

\begin{abstract}
With rapid advancements in large language models (LLMs), AI-generated content (AIGC) has emerged as a key driver of technological innovation and economic transformation. Personalizing AIGC services to meet individual user demands is essential but challenging for AIGC service providers (ASPs) due to the subjective and complex demands of mobile users (MUs), as well as the computational and communication resource constraints faced by ASPs. To tackle these challenges, we first develop a novel multi-dimensional quality-of-experience (QoE) metric. This metric comprehensively evaluates AIGC services by integrating accuracy, token count, and timeliness. We focus on a mobile edge computing (MEC)-enabled AIGC network, consisting of multiple ASPs deploying differentiated AIGC models on edge servers and multiple MUs with heterogeneous QoE requirements requesting AIGC services from ASPs. To incentivize ASPs to provide personalized AIGC services under MEC resource constraints, we propose a QoE-driven incentive mechanism. We formulate the problem as an equilibrium problem with equilibrium constraints (EPEC), where MUs as leaders determine rewards, while ASPs as followers optimize resource allocation. To solve this, we develop a dual-perturbation reward optimization algorithm, reducing the implementation complexity of adaptive pricing. Experimental results demonstrate that our proposed mechanism achieves a reduction of approximately $64.9\%$ in average computational and communication overhead, while the average service cost for MUs and the resource consumption of ASPs decrease by $66.5\%$ and $76.8\%$, respectively, compared to state-of-the-art benchmarks.

\end{abstract}


\section{Introduction}\label{I}
\subsection{Background and Motivation}
Advancements in machine learning algorithms and large language models (LLMs) have significantly enhanced the quality and diversity of AI-generated content (AIGC), driving innovation across various domains and transforming the economic landscape~\cite{xu2023unleashing}. AIGC services, such as OpenAI's ChatGPT, have demonstrated remarkable capabilities in tasks like language understanding, text generation, summarization, and conversational AI.
With the widespread adoption of AIGC services, users across diverse applications have increasingly varied expectations for service performance, making personalized service a critical requirement. However, providing personalized AIGC services is hindered by the subjective and complex demands of mobile users (MUs), as well as the limited computational and communication resources of AIGC service providers (ASPs). Addressing these barriers in mobile AIGC networks requires overcoming both technical and economic challenges.

\textbf{Challenge 1: Resource Constraints}.
Requesting real-time AIGC services directly from the cloud often incurs significant latency and network congestion. Furthermore, the inference of LLMs necessary for AIGC services is resource-intensive and time-consuming, making it impractical for deployment on mobile devices due to high costs. To this end, current schemes~\cite{10409284, xu2023sparks} deploy pre-trained AIGC models on edge servers, leveraging the infrastructure of wireless edge networks. When MUs request AIGC services, they can send the prompts to edge servers. The prompts are then executed
using generative AI models, and the results are sent back to MUs. This approach alleviates the computational burden on mobile devices and enables flexible and scalable AIGC service provisioning in mobile edge networks.
 However, when numerous MUs simultaneously request services, edge servers with limited computational and communication resources can become overwhelmed, resulting in increased latency and degraded performance.

   \textbf{Challenge 2: Personalized Demands}. 
MUs exhibit diverse service demands depending on their applications. For example, MUs in vehicular networks prioritize concise and rapid decision-making responses~\cite{vehicle_concise}, whereas those involved in academic essay writing require detailed and profound content. To meet MUs' expectations, aligning AIGC services with their personalized demands is essential, considering the unique characteristics of different application areas.
However, achieving this goal faces two main problems. First, the stochastic nature of LLMs, which relies on next-token prediction~\cite{xu2024large}, may result in incoherent outputs, factual inaccuracies, or contradictions. This uncertainty makes it difficult for LLMs to consistently generate highly personalized content that accurately matches MUs' expectations. Second, existing service assessment metrics, such as designer-perceivable quality of experience (QoE)~\cite{10330096}, joint perpetual similarity and quality (JPSQ)~\cite{optimal_B2}, and user-perceived quality~\cite{du2023enabling}, fail to adequately consider the complex trade-offs between accuracy, token count, and timeliness in mobile AIGC networks.

 \textbf{Challenge 3: Self-Interested Behaviors of ASPs and MUs}.
ASPs are driven by profit and tend to reduce resource investment without sufficient incentives. Meanwhile, MUs prefer receiving high-quality services at minimal cost. This misalignment of goals can lead to a structural imbalance between service provision and economic incentives.
In response to this contradiction, existing research has proposed incentive strategies using mechanisms such as multi-agent deep reinforcement learning~\cite{optimal_B1}, game-theoretic~\cite{unified}, and contract theory~\cite{10158526,10529221}. However, these studies fail to explicitly model the multi-dimensional demands of individual users, making it challenging for ASPs to optimize resource allocation for diverse user needs. Furthermore, they often focus on a single type of resource and overlook the complexity involved in designing incentives for multi-ASP scenarios and the inherent competition among MUs.

To address these challenges, we propose a dual analytical framework from the perspectives of both MUs and ASPs, leading to the following two key research questions:

i) \textit{How can MUs optimize service costs while obtaining personalized AIGC services that meet their heterogeneous requirements in a resource-competitive network?}
   
ii) \textit{How can ASPs optimally allocate computational and communication resources to minimize resource consumption while delivering personalized AIGC services?}

\begin{figure}[t]
	\centering
	\includegraphics[width=\linewidth]{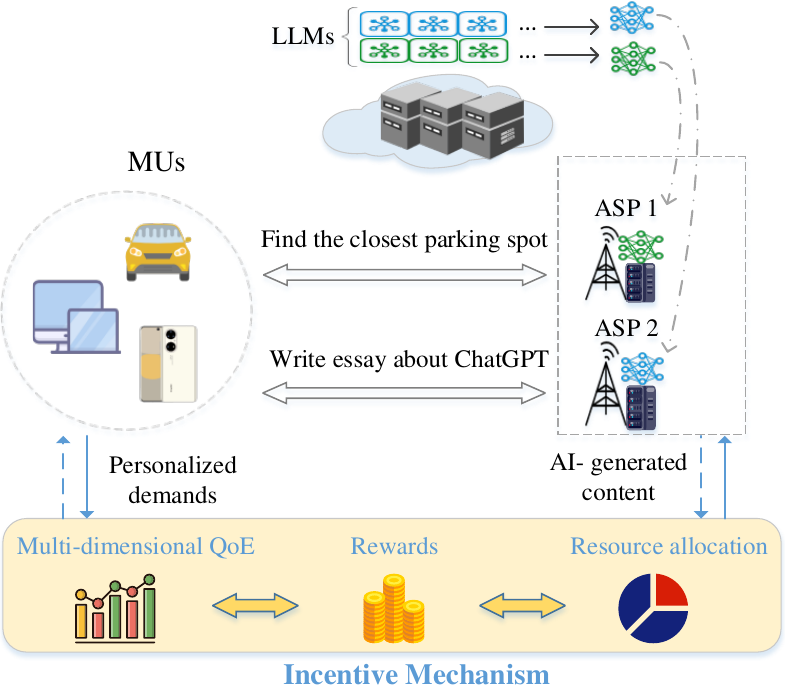}
	\caption{Illustration of the proposed QoE-driven incentive mechanism in an MEC-enabled AIGC network. The network comprises multiple ASPs deploying differentiated AIGC models on edge servers, while multiple MUs with heterogeneous QoE requirements request services from these ASPs.}
	\label{ex}
\end{figure}

\subsection{Solution and Contributions}
As illustrated in Fig.~\ref{ex}, we consider a mobile edge computing (MEC)-enabled AIGC
network, consisting of multiple ASPs with differentiated AIGC
models deployed on edge servers and multiple MUs with
heterogeneous QoE requirements requesting AIGC services
from ASPs. To address the above problems, we first design a multi-dimensional quality of experience (QoE) metric to quantify MUs' personalized demands. 
Building on this metric, we propose a QoE-driven incentive mechanism that aligns the interests of both MUs and ASPs.
 In our framework, MUs act as leaders by driving the trading process through differentiated reward strategies, while ASPs function as followers by dynamically adjusting resource allocation in response. This mechanism achieves equilibrium through bidirectional optimization: Competing MUs strategically determine rewards in a non-cooperative game, each aiming to obtain better personalized AIGC services at lower costs (problem i)), while ASPs jointly optimize the allocation of computational and communication resources to balance QoE fulfillment and resource consumption based on these rewards (problem ii)).
In summary, the key contributions of this work are summarized as follows. 

\begin{itemize}
\item \textit{Incentive mechanism design for on-demand service provisioning in resource-constrained AIGC edge networks.} This mechanism integrates the perspectives of MUs and ASPs, optimizing service costs for MUs to access personalized AIGC services while ensuring that ASPs can efficiently deliver on-demand services without over-provisioning resources. To our knowledge, this is the first analytical study on incentive mechanism design explicitly tailored to the unique characteristics of LLM-based text services in multi-MU and multi-ASP scenarios.

   \item \textit{Novel design of multi-dimensional QoE metric to quantify service quality.} To quantify MUs' personalized demands, we design a metric integrating three critical dimensions: accuracy, token count, and timeliness. The accuracy is evaluated through CoT reasoning performance gaps, capturing discrepancies between the expected and actual model outputs. Token count is determined by the total number of input and output tokens, while timeliness is constrained by the maximum tolerable latency. Through comprehensive analysis and validation of CoT reasoning steps and examples, the proposed QoE metric effectively models diverse service requirements while accounting for the limited resources of ASPs, supporting on-demand service provision.

  \item \textit{Theoretical analysis and algorithm designs.} We propose a two-level game-theoretic model to analyze the MU's optimal reward and the ASP's optimal resource allocation. Given multiple ASPs with differentiated AIGC models deployed on edge servers and multiple MUs with heterogeneous QoE requirements, the optimization problems of the ASP and the MU are of high complexity (e.g., inter-user competition and resource constraints).  We solve these problems by reformulating them as an equilibrium problem with equilibrium constraints (EPEC) framework, theoretically proving the existence and uniqueness of the equilibrium, and designing a dual-perturbation
reward optimization algorithm to reduce the implementation
complexity of adaptive pricing.

     \item \textit{Performance evaluation}. Numerical results based on real-world scenarios show that our proposed mechanism achieves approximately $64.9\%$ reductions in average computational and communication overhead, $66.5\%$ decreases in average MU service cost, and $76.8\%$ savings in resource consumption of ASPs on average, compared to state-of-the-art benchmarks.

\end{itemize}
The rest of this paper is organized as follows.
Section~\ref{II} discusses related work. In Section~\ref{III}, we present
the system overview and design a multi-dimensional QoE metric. Section~\ref{IV} gives the game formulation, and Section~\ref{V} analyzes the existence of the equilibria at two levels. Section~\ref{exp} shows numerical experiments to demonstrate the advantages of the proposed incentive mechanism, and finally, Section~\ref{VII} concludes the paper.

\section{Related Work}\label{II}
In this section, we discuss fundamental research in three areas: AIGC services in networks, multimodal prompt engineering, and incentive mechanisms for AIGC services. First, we explore issues and methods for implementing AIGC in resource-constrained edge environments, laying the foundation for our study. Building on this, we investigate multi-modal prompt engineering techniques essential for quantitatively assessing content accuracy in the QoE design. Finally, we highlight our unique contributions to incentive mechanism designs, addressing critical gaps in existing studies.

\subsection{AIGC Services in Networks}

To alleviate the computational burden on the cloud and reduce high latency caused by long transmission distances, Du~\textit{et al.}~\cite{10409284,du2023enabling} proposed an AIGC-as-a-Service (AaaS) architecture and introduced a dynamic ASP selection scheme for optimal user-provider connections.

To ensure secure and trustworthy AIGC services, Lin \textit{et al.}~\cite{10483549} designed a decentralized trust mechanism using smart contract-based verification to prevent unreliable outcomes in the Metaverse. Similarly, Liu \textit{et al.}~\cite{10504615} proposed a systematic approach incorporating MASP selection, payment schemes, and fee-ownership transfer, presenting a blockchain framework to protect mobile AIGC services.
For collaborative inference, several studies~\cite{10172151,10630955,10437617} explored distributed task allocation frameworks. These works achieved efficient AIGC services by optimizing workload distribution across multiple mobile devices, edge servers, and distant data centers.
Moreover, the authors in \cite{10746594}, \cite{optimal_B2}, and \cite{10273254} explored the integration of semantic communication and AIGC services, focusing on optimizing bandwidth and resource utilization while enhancing content quality in wireless networks.
In addition to the above works, several studies focus on performance improvement in edge networks. For instance, Wang \textit{et al.}~\cite{10515205} proposed the WP-AIGC framework, integrating wireless perception and AIGC to optimize computing resources at the edge server. Wang \textit{et al.}~\cite{optimal-f1} proposed a joint optimization scheme for offloading, computation time, and diffusion steps in the reverse diffusion stage, using average error as a metric to evaluate the quality of generated results.

These works, from diverse perspectives, have significantly advanced the deployment and optimization of AIGC services in networks, providing valuable insights for future research. Notably, distinct from existing studies, our work focuses on designing an incentive mechanism that motivates ASPs to provide services on-demand, meeting MUs' personalized demands in multi-user and multi-ASP scenarios.

\subsection{Multimodal Prompt Engineering} 
Prompt engineering in large language models (LLMs) involves designing and optimizing prompts or instructions to guide the model's output generation~\cite{xu2024generative}. This technique enhances LLM performance across diverse tasks, including language modeling, question answering, and image captioning. Wei~\textit{et al.}~\cite{wei2022chain} first demonstrated how generating a Chain-of-Thought (CoT) improves LLMs' ability to perform complex reasoning tasks. Their approach, which uses a few CoT demonstrations as exemplars in prompts, significantly improved performance on arithmetic, commonsense, and symbolic reasoning tasks. Building upon this, Wang~\textit{et al.}~\cite{wang2023plan} proposed the Plan-and-Solve (PS) prompting strategy for multi-step reasoning tasks. By dividing a task into smaller subtasks and executing them sequentially, PS prompting addresses missing-step errors and enhances the quality of generated reasoning steps.
Incorporating multimodal reasoning further enriches LLM capabilities, enabling the synthesis of consistent and coherent CoTs across different modalities. Lu~\textit{et al.}~\cite{lu2022learn} introduced ScienceQA, a multimodal reasoning benchmark comprising approximately 21k multiple-choice questions on diverse science topics. 
Similarly, Zhang~\textit{et al.} in \cite{zhang2023multimodal} proposed a new approach called Multimodal-CoT that incorporates both language (text) and vision (images) modalities into a two-stage framework for CoT prompting. 

While these studies have advanced CoT prompting and multimodal reasoning, they have not studied how to improve the quality of LLM-based services by capturing the user's demands from the theoretical properties of CoT and reducing resource consumption from the service and optimization perspectives.

\begin{table}[t]
\caption{Incentive Mechanisms for AIGC Services.}
\label{realted}
\begin{tabular}{c|c|c|c|c}
\hline
Ref.                                                & \begin{tabular}[c]{@{}c@{}}Multi-MU\\ Multi-ASP\end{tabular} & \begin{tabular}[c]{@{}c@{}}Optimization\\objective\end{tabular}                                                                  & New metric                                                                       & \begin{tabular}[c]{@{}c@{}}On \\ demand\end{tabular} \\ \hline
{\cite{optimal_B1}}                                             & \checkmark                                 & \begin{tabular}[c]{@{}c@{}}Service \\allocation\end{tabular}                                                                   & $\times$                                                                         & $\times$                                             \\ \hline
{\cite{unified}}                                            & $\times$                                                  & \begin{tabular}[c]{@{}c@{}} Service\\provision    \end{tabular}                                                                & \begin{tabular}[c]{@{}c@{}}Blind image \\spatial quality \end{tabular} & $\times$                                             \\ \hline
{\cite{10158526}}                                            & \checkmark                                 & \begin{tabular}[c]{@{}c@{}}Information\\ sharing\end{tabular}                                                            & \begin{tabular}[c]{@{}c@{}}Bit error probability  \end{tabular}                                                           & $\times$                                             \\ \hline
{\cite{10529221}}                                            & \checkmark                                                   & \begin{tabular}[c]{@{}c@{}}Service \\ cost  \end{tabular}                                                                     & $\times$                                                                         & $\times$                                             \\ \hline
{\cite{IMFL}}                                            & $\times$                                                  & \begin{tabular}[c]{@{}c@{}}Service\\ performance \end{tabular}                                                                   & Data quality                                                                     & $\times$                                             \\ \hline
{\cite{zhan2024vision}}                                            & \checkmark                                 & \begin{tabular}[c]{@{}c@{}}Task\\ allocation    \end{tabular}                                                                  & \begin{tabular}[c]{@{}c@{}}  Difficulty\\assessment\end{tabular}   & $\times$                                             \\ \hline
{\cite{Wen}}                                            & $\times$                                                  & \begin{tabular}[c]{@{}c@{}}Service \\ cost  \end{tabular}                                                                   & $\times$                                                                         & $\times$                                             \\ \hline
{\cite{10707303}}                                            & $\times$                                                  & \begin{tabular}[c]{@{}c@{}}Service \\performance\end{tabular} & \begin{tabular}[c]{@{}c@{}}Quality and latency \\of image generation    \end{tabular}                                                             & $\times$                                             \\ \hline
{\cite{xu2024cached}}                                            & \checkmark                                 & \begin{tabular}[c]{@{}c@{}}Service \\performance\end{tabular}         & Age 
of thought                                                                              & $\times$                                             \\ \hline
\begin{tabular}[c]{@{}c@{}}Our\\  work\end{tabular} & \checkmark                                 & \begin{tabular}[c]{@{}c@{}}Service \\performance\end{tabular}              & QoE                                                                              & \checkmark                            \\ \hline
\end{tabular}
\end{table}

\subsection{Incentive Mechanisms for AIGC Services}
Mechanisms incentivizing MU participation are critical in scenarios like federated learning (FL) and data trading~\cite{10632200}. For example, Chen~\textit{et al.}~\cite{IMFL} proposed a data quality assessment method for AIGC-generated data samples and designed a reward mechanism to incentivize MUs to participate in FL. Du~\textit{et al.}~\cite{10158526} developed an incentive mechanism based on generative AI and contract theory to encourage MUs to share semantic information. Zhan~\textit{et al.}~\cite{zhan2024vision} designed an AIGC task assignment framework for automated difficulty assessment using a visual language model (VLM) with the aid of contract theory. Fan~\textit{et al.}~\cite{optimal_B1} developed a decentralized incentive mechanism for mobile AIGC service allocation that balances service provision with MU demand. These studies predominantly adopt server-driven pricing models to incentivize MU participation in service provision.

In contrast to mechanisms focused on MU participation, another research direction, similar to our work, focuses on incentivizing ASPs to allocate resources and provide AIGC services. Wen~\textit{et al.}~\cite{Wen} proposed a user-centric incentive mechanism to motivate ASPs to provide AIGC services while incorporating prospect theory to model the subjective utility of clients. Du~\textit{et al.}~\cite{10529221}
extended the scope to multi-MU and multi-ASP scenarios employing contract theory to reward ASPs based on their resource contributions. Although these mechanisms improve flexibility and applicability, they focus primarily on computational resources and overlook joint resource optimization, as well as the unique requirements of AIGC services.
To address the need for AIGC-specific metrics, researchers have proposed mechanisms targeting particular applications. For example, Wang~\textit{et al.}~\cite{unified} introduced metrics to evaluate the precision of wireless perception and the quality of AIGC, linking resource allocation to service quality through a diffusion-based pricing strategy. Ye~\textit{et al.}~\cite{10707303} modeled the relationship between prompt optimization, diffusion denoising steps, and AIGC quality, designing a quality-based contract to improve AIGC generation while reducing latency.
While these works improve the quality of service in AIGC, they are limited to single-MU scenarios and primarily focus on image generation. In addition, the authors in~\cite{xu2024cached} proposed the age-of-thought (AoT) metric and a deep Q-network-based auction to incentivize the provisioning of LLM agents, but focused on cached resources without addressing on-demand services for personalized demands.

 Unlike the above works, our study focuses on AIGC services based on LLMs in multi-MU and multi-ASP scenarios. We design a comprehensive QoE metric that captures MUs' personalized demands across multiple dimensions, tailored to the unique characteristics of text-based content generation. By jointly optimizing computational and communication resources, our mechanism incentivizes ASPs to allocate resources on demand, enabling scalable, efficient, and personalized service provision. A detailed comparison of these approaches is provided in Table~\ref{realted}, highlighting the key contributions of this study.
\begin{table}[t]
    \centering
    \caption{Summary of Main Notations.}
    \label{parameters}
    \begin{tabular}{c||l}
        \hline
        \textbf{Notation} & \textbf{Definition} \\
        \hline \hline
        $n,N$ & Index of an ASP, number of ASPs \\
        \hline
        $m,M$ & Index of an MU, number of MUs \\ 
        \hline
        $a_{0}^{nm}$ & Input message or task provided by MU $m$ to ASP $n$ \\
        \hline
         $a_{i}^{nm}$ & Message chain provided by ASP $n$ to MU $m$ \\
        \hline
        $B_{nm}$ & ASP $n$'s communication resource allocated for MU $m$ \\
        \hline
         $c_{n}^{f}$ & Cost factor per unit of computational resources\\ 
        \hline
        $c_{n}^{B}$ & Cost factor per unit of communication resources\\
        \hline
        $f_{nm}$ & ASP $n$'s computational resource allocated for MU $m$ \\
        \hline
        $K$ & Number of CoT examples \\
        \hline
        $O_{k}^{nm}$ & CoT examples utilized by ASP $n$ in providing the \\ & service to MU $m$ \\       
        \hline
        $\mathcal{Q}_{nm}$ & Quantified QoE value for MU $m$ provided by ASP $n$ \\
        \hline
        $R_{nm}$ & Reward for unit QoE value offered by MU $m$ to ASP $n$ \\
        \hline
        $x^{\text{in}}_{nm}$ & Input token from MU $m$ to ASP $n$ \\
        \hline
        $x^{\text{out}}_{nm}$ & Output token requested from ASP $n$  by MU $m$ \\
        \hline
         $\gamma_{nm}$ & SNR for the communication between MU $m$ and ASP $n$\\
        \hline
        $\hat{\theta}_{nm}$ & Upper bound on the performance gap of the service \\ & provided by ASP $n$ to MU $m$ \\
        \hline
        $\Upsilon(c_{n}^{*})$ & Maximum deviation between the context $c_{n}$ owned by \\ & ASP $n$ and the distribution of true contexts\\
        \hline
         $\zeta(a_{i}^{nm})$ & Ambiguity of chain $(a_{i}^{nm})_{0\le i\le l}$ \\ \hline
        $\mu_{m}$ & Profit conversion factor for MU $m$'s AIGC services \\
        \hline
        $\kappa_{n}$ & Maximum tolerable latency for ASP 
$n$'s service. \\
        \hline
         $\xi_{n}$ & Computational resource required per token of ASP $n$\\
        \hline
    \end{tabular}
\end{table}
\begin{figure*}[t]
	\centering
	\includegraphics[width=\linewidth]{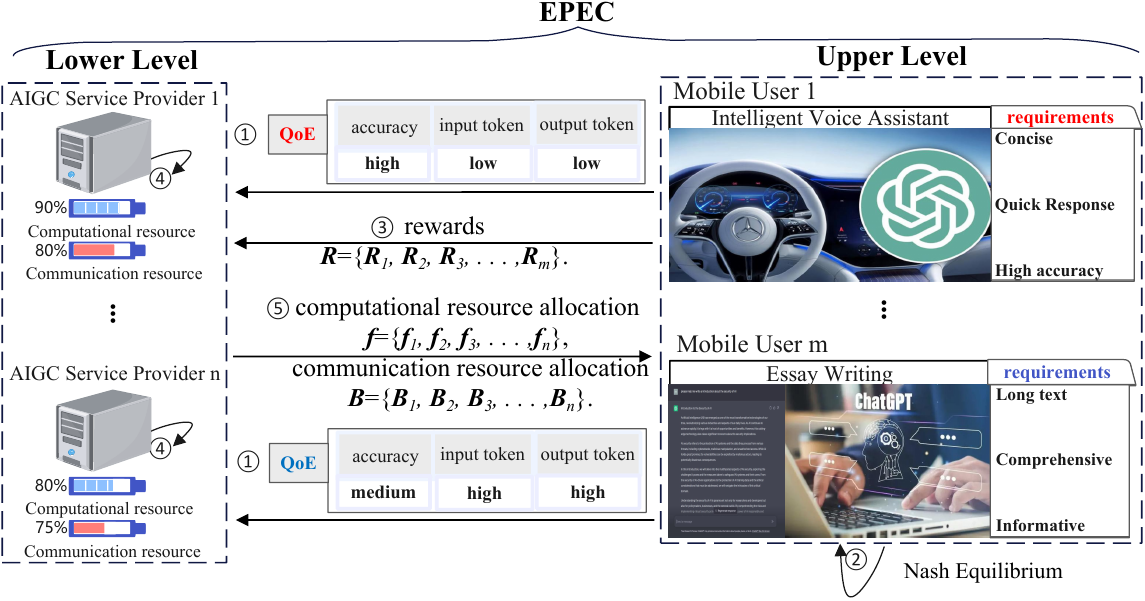}
	\caption{The workflow of the QoE-driven incentive mechanism. Each MU has distinct service requirements reflecting individual preferences. For example, vehicles demand concise and accurate responses to ensure safety, whereas content creators require richer inputs and outputs to support content generation.}
	\label{flow}
\end{figure*}
\section{System Overview}\label{III}
We first present the system model of the incentive mechanism for personalized AIGC services in Section~\ref{iii-A}. We then discuss the ambiguity of language employed in LLMs in Section~\ref{lau}. Following this, we design a novel QoE metric to model MU demands across three dimensions in Section~\ref{QoEnovel}. 
 To simplify the presentation, we summarize the essential notations in Table~\ref{parameters}.

\subsection{System Model}\label{iii-A}
 We consider a scenario involving a set $\mathcal{M} \textrm{=} \left \{ 1,\ldots,M \right \}$  of MUs and a set $\mathcal{N} \textrm{=} \left \{ 1,\ldots,N \right \}$ of ASPs. As depicted in Fig.~\ref{flow}, each MU has unique service requirements that reflect  individual preferences. The requirements for AIGC services vary according to specific applications. For example, in vehicular environments, MUs aim for fast, accurate, and concise responses from AIGC services to enhance safety and convenience. Conversely, when engaged in content creation tasks, MUs prioritize content quality, depth, and relevance over response time. ASPs are resource-limited edge servers equipped with different types of LLMs. For example, ASP 1 is dedicated to vehicular networking, while ASP 2 specializes in academic writing. The objective is to maximize the respective benefits of MUs and ASPs in personalized AIGC trading networks through a flexible user-driven rewarding mechanism and ASPs’ on-demand resource allocation.
 In this mechanism, MUs act as upper-level leaders determining rewards, while ASPs serve as lower-level followers making resource allocation decisions. This structure ensures the provision of personalized AIGC services that effectively meet the diverse demands of MUs.
First, MUs broadcast their specific demands (accuracy, input and output tokens) defined in QoE to ASPs (\ding{172}). Subsequently, the MUs determine the reward per QoE given to ASPs based on a Nash equilibrium (\ding{173}), with the initial rewards set as random values, and transmit them to ASPs (\ding{174}). The ASPs then make resource allocation decisions based on the QoE, cost, and rewards given by the MUs (\ding{175}). Based on the feedback from ASPs (\ding{176}), MUs adjust the rewards to maximize their benefits (\ding{173}). Steps \ding{173} and \ding{176} are repeated until an agreement is reached between MUs and ASPs.
The iterative process ensures mutual benefits and ultimately enhances the AIGC service experience for all MUs.

\begin{figure*}[t]
    \centering
\includegraphics[width=1\linewidth]{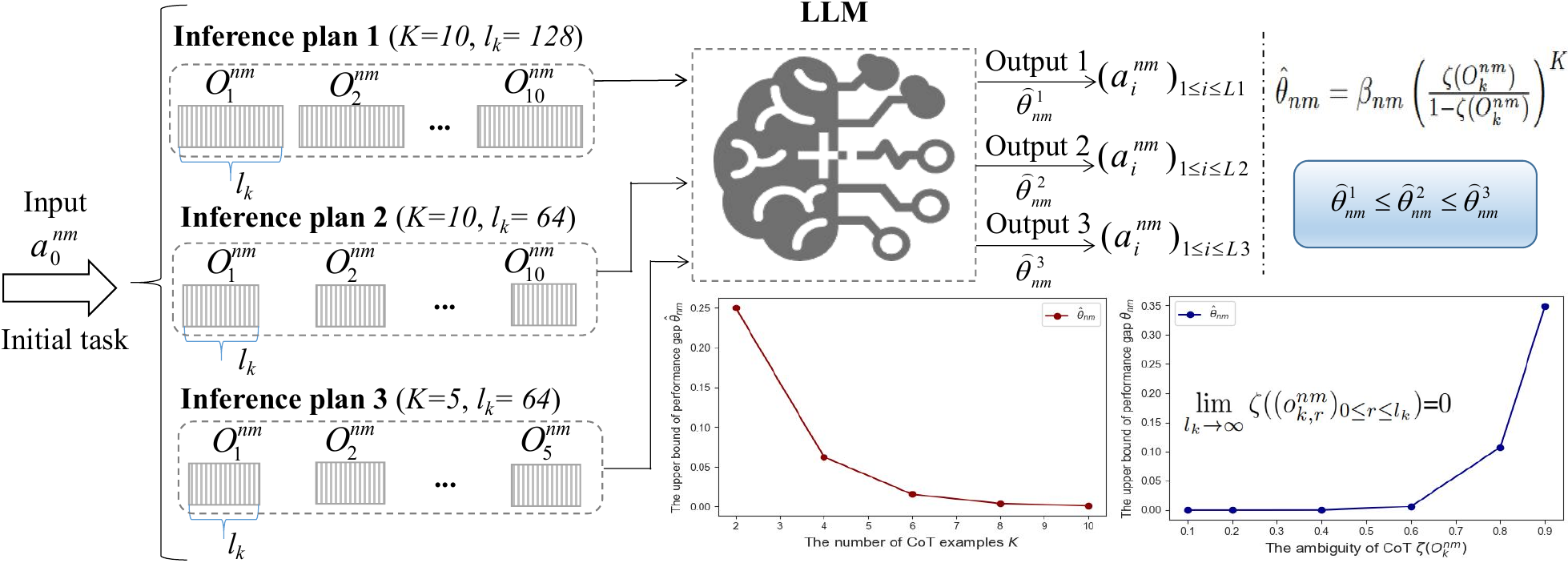}
    \caption{Impact of different CoT examples on performance gap $\hat{\theta}_{nm}$.}
    \label{fig:cot}
\end{figure*}

\subsection{Ambiguity of Language}\label{lau}
Chain-of-Thought (CoT) prompting~\cite{tutunov2023large} is a technique designed to enhance the reasoning capabilities of LLMs. By guiding the model to break down complex problems into logical steps, CoT prompts enable incremental solution derivation that improves contextual understanding and logical consistency. CoT examples play a crucial role in this process. They are a set of exemplary reasoning chains that demonstrate how the model can complete a specific task through step-by-step reasoning. These examples not only help the model learn how to structure its reasoning process but also guide it in generating more accurate and consistent outputs.
However, the inherent ambiguity in natural language poses a significant challenge to LLMs' reasoning abilities. This ambiguity affects their ability to determine precise intentions or contexts behind ambiguous messages or sequences. 

To investigate the impact of this ambiguity on inference performance, we explore and analyze the theoretical foundations underlying CoT prompting in the inference process.
   When serving MU $m$, ASP $n$ is provided with $K$ CoT examples of varying lengths, denoted as $O_{k}^{nm}=(o_{k,r}^{nm})_{0\le r\le l_{k}}$, where $l_{k}$ represents the length of the chain $O_{k}^{nm}$, and each $o_{k,r}^{nm}$ is a sequence of tokens representing a thought or reasoning step. These examples are designed to help ASP $n$ produce correct answers via CoT generation for MU $m$. For all $k$, $O_{k}^{nm}$ are generated with true intentions $(I^{nm})^{*}$ and context $c_{n}^{*}$. Following this, ASP $n$ generates a series of messages denoted as $(a_{i}^{nm})_{0\le i\le l}$ based on the provided CoT examples $O_{k}^{nm}$ and the initial task $a_{0}^{nm}$.
Based on the CoT prompting process, we can adapt the following definition of ambiguity from~\cite{jiang2023latent}.

\begin{myDef}
\textbf{(The ambiguity of the chain)} When ASP $n$  performs the reasoning task for MU $m$, the ambiguity of a message chain $(a_{i}^{nm})_{0\le i\le l}$, derived from true intentions $(I^{nm})^{*}$ and true context $c_{n}^{*}$, is defined as the complement of the likelihood $\hat{q}$ of the context $c_{n}^{*}$
  and intentions $(I^{nm})^{*}$
  conditioned on $(a_{i}^{nm})_{0\le i\le l}$,
i.e., 
\begin{equation}
   \zeta((a_{i}^{nm})_{0\le i\le l})=1-\hat{q}(c_{n}^{*},(I^{nm})^{*}|(a_{i}^{nm})_{0\le i\le l}).
\end{equation}
\end{myDef}

Similarly, the ambiguity of the message chain $\zeta ((O_{k}^{nm})_{1\le k\le K})$ can be defined as the complement of the likelihood of the context $c_{n}^{*}$
  and intentions $(I^{nm})^{*}$
  conditioned on $(O_{k}^{nm})_{1\le k\le K}$,
i.e., 
\begin{equation}
 \zeta((O_{k}^{nm})_{1\le k\le K})=1-\hat{q}(c_{n}^{*},(I^{nm})^{*}|(O_{k}^{nm})_{1\le k\le K}).
\end{equation}

After defining the ambiguity of the chain, we focus on evaluating the overall performance of the model using CoT prompting technology.
To achieve this, we use the concept of ambiguity to assess the model's performance by comparing the likelihood of the results obtained through CoT examples with those derived from the true natural language distribution. This comparison quantifies the performance gap between the service provided by ASP $n$ for MU 
$m$ and the ideal result for MU 
$m$, denoted as
\begin{equation}
\begin{aligned}
\theta_{nm} &= \left| p_{n}\left( (a_{i}^{nm})_{1\le i\le l}|a_{0}^{nm},O_{k}^{nm} \right) \right. \\
& \quad - \left. \hat{q}\left( (a_{i}^{nm})_{1\le i\le l}|a_{0}^{nm},c_{n}^{*} \right) \right|.
\end{aligned}
\end{equation}

To further understand this performance gap, we explore a theoretical framework that considers the distribution of contexts in the training datasets, aiming to gain insights into the upper bound of the gap. 

\begin{myTheo}\label{th1}
    \textbf{(The upper bound of the performance gap)} When ASP $n$ infers the task of MU $m$ based on CoT prompting technology, consider a set of 
$K$ CoT examples  $O_{k}^{nm} = (o_{k,r}^{nm})_{0 \leq r \leq l_k}$, generated from the intention $(I^{nm})^{*}$
with the optimal context $c_{n}^{*} \sim q(c_{n})$. Then, for any message sequence $(a_{i}^{nm})_{1 \leq i \leq l}$, we have~\cite{tutunov2023large}: 
\begin{equation}
  \theta_{nm}\le 
  2 \frac{(\Upsilon(c_{n}^{*}))^{K} \zeta(a_{0}^{nm})}{1-\zeta (a_{0}^{nm})}\prod_{k=1}^{K}\frac{\zeta(O_{k}^{nm})}{1-\zeta (O_{k}^{nm})},
\end{equation}where $a^{nm}_{0}$, sampled from $q(\cdot|(I^{nm}_{0})^{*})$, is the input message or task generated from $(I^{nm}_{0})^{*}$, which is sampled from $q(\cdot|c_{n}^*)$.
 $\Upsilon(c_{n}^{*})=\sup_{c \in C} \frac{\hat{q}(c_{n}^{*})}{\hat{q}(c_{n})}$ is  a  skewness parameter, indicating the maximum deviation between the context $c_n$ owned by ASP $n$ and the distribution of true contexts. 
\end{myTheo}

If $\Upsilon(c_{n}^{*})=1$, it implies that natural language does not discriminate among specific contexts. This condition holds when dealing with sufficiently large and well-balanced datasets~\cite{liu2022selfsupervised,LLM_uniform}. 
Moreover, under certain conditions,  a geometrical convergence rate can be obtained as follows.

\begin{myLemm}\label{lemma1}
 Consider a set of CoT examples $O_{k}^{nm} = (o_{k,r}^{nm})_{1\le r \le l_{k}} $ satisfying the following conditions~\cite{xu2024cached}:
 
1) The CoT examples $O_{k}^{nm}$ are carefully selected such that the true $c_{n}^{*}$ can be recovered from $O_{k}^{nm}$ with relatively high certainty. 

2) As the length of the sequence grows, the associated ambiguity measure $\zeta(O_{k}^{nm})$ diminishes.  

Then, for any fixed $\eta \in [0,\frac{1}{2})$, there exists a length threshold $l_{k,\eta }^{*} \in \mathbb{N}$ , such that for any $l_{k} \ge l_{k,\eta }^{*}$, we can have:
\begin{equation}
    \zeta(O^{nm}_{k})\le \eta .
\end{equation}
\end{myLemm} 

\begin{proof}
In practice, satisfying condition 1) is challenging as there is no rigorous procedure to quantify ambiguity for a given sequence of thoughts. However, under condition 2), we can obtain $\lim\limits_{l_{k} \to \infty} \zeta((o_{k,r}^{nm})_{0\le r\le l_{k}})\textrm{=}0$. Hence, there exists $l_{k, \eta} \in \mathbb{N}$ such that for any $l_k \ge l_{k, \eta}^{*}$, we have $\zeta(O^{nm}_{k})\le \eta$.
\end{proof}
 As the number of CoT examples satisfying the conditions increases, a geometrical convergence rate for the gap $\theta_{nm}$ can be established as 
 \begin{equation}\label{boundused}
    \theta_{nm}\le 
  \beta_{nm}(\frac{\eta}{1-\eta})^{K},
\end{equation}where $\beta_{nm}=2 \frac{(\Upsilon(c_{n}^{*}) )^{K}\zeta(a_{0}^{nm})}{1-\zeta (a_{0}^{nm})}$  and $\frac{\eta}{1-\eta} \in [0,1)$.

According to the analysis of \textbf{Theorem ~\ref{th1}} and \textbf{Lemma~\ref{lemma1}}, the performance gap is primarily influenced by the number of CoT examples $K$ and the reasoning step length $l_{k}$. As illustrated in Fig.~\ref{fig:cot}, there is a clear functional relationship between the number of CoT examples 
$K$ and the performance gap 
$\theta$, demonstrating that an increased number of CoT examples significantly reduces the performance gap, leading to higher accuracy in AIGC.  
In particular, $\theta\rightarrow 0$ indicates high-quality AIGC, while $\theta\rightarrow \infty$ implies a significant performance gap and low accuracy. However, the impact of the reasoning step length 
$l_{k}$ on the performance gap does not exhibit a clear functional relationship. To further validate the theoretical analysis, we conducted experiments on an open platform to empirically examine how reasoning step length affects the accuracy of inference results in real-world scenarios.
As illustrated in Fig.~\ref{fig:MPEng}, the ASP can choose CoT examples with different reasoning step lengths for approaching the final answer\footnote{ The experiments are conducted on the ScienceQA benchmark~\cite{zhang2023multimodal}. This large-scale multimodal science question dataset contains 21k multiple-choice questions across 3 subjects, 26 topics, 127 categories, and 379 skills, divided into training, validation, and test splits with 12726, 4241, and 4241 examples, respectively. This setting aimed to demonstrate the efficacy of multimodal-CoT by leveraging both language and vision modalities in a two-stage reasoning framework.}. The results indicate that as the reasoning step length increases, the associated ambiguity measure $\zeta(O_{k}^{nm})$ decreases, leading to a smaller $\hat{\theta}_{nm}$ and thus reducing the performance gap. In conclusion, when using CoT with a larger number of steps, LLMs often produce more accurate results, as revealed in \textbf{Lemma~\ref{lemma1}} and literature~\cite{jin2024impact}.

\begin{figure*}[t]
    \centering
\includegraphics[width=1\linewidth]{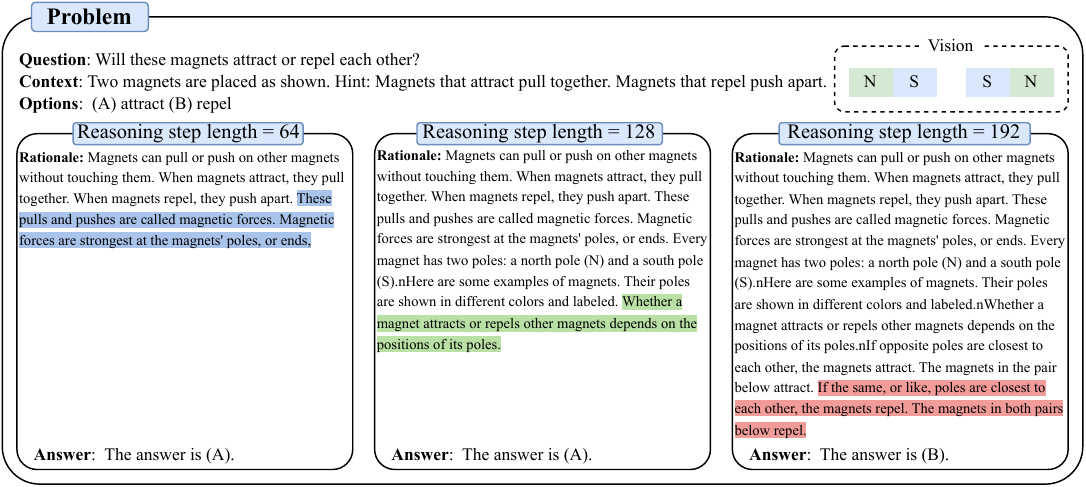}
    \caption{Example of different reasoning step lengths in CoT. When the number of CoT examples is fixed, using longer reasoning steps within each example generally improves accuracy. ASPs can therefore select CoT examples with varying step lengths to better approach the final answer.}
    \label{fig:MPEng}
\end{figure*}

\subsection{QoE of AIGC Services}\label{QoEnovel}
AIGC services exhibit a complex interdependence among accuracy, token count, and service latency. Due to the autoregressive nature of LLMs, computational demand scales quadratically with the number of tokens, as self-attention requires each token to compute relationships with all preceding tokens. Furthermore, improving accuracy incurs growing resource costs that further amplify latency. Compared to traditional services, AIGC faces more intricate multi-dimensional trade-offs, requiring an effective balance between cost and performance.

To this end, we propose a novel QoE metric specifically designed for AIGC services, integrating three key dimensions: accuracy, token count, and response timeliness. By explicitly modeling the trade-offs among these factors, this metric provides a unified framework for evaluating both technical performance and user experience. 
We utilize the derived upper bound $\hat{\theta}_{nm}= \beta_{nm}(\frac{\eta}{1-\eta})^{K}, \ n \in \mathcal{N}, \ m \in \mathcal{M}$ to evaluate the output accuracy. As demonstrated in Section~\ref{lau}, a smaller $\hat{\theta}_{nm}$ indicates that the model's output is closer to the ideal result, implying higher accuracy. Achieving  higher accuracy typically requires more computational effort, and the computational demand for inference is influenced by the number of tokens processed. 

Moreover, ASP $n$ needs to allocate computational resource $\bm{f}_{n}\triangleq\left(f_{n1}, f_{n2}, \ldots, f_{nm}\right)$ and bandwidth $\bm{B}_{n}\triangleq\left( B_{n1}, B_{n2}, \ldots, B_{nm}\right)$ for content generation and transmission. Considering these factors jointly, the QoE of the AIGC service provided by ASP $n$ to MU $m$ is defined as
\begin{equation}
\begin{aligned}
    \mathcal{Q}_{nm} = \kappa_{n} &- \frac{\xi_{n} \ln(1/\hat{\theta}_{nm}) \sum_{i=0}^{x^{\text{out}}_{nm}-1} \left( x^{\text{in}}_{nm} + i \right)}{f_{nm}} \\&- \frac{32 (x_{nm}^{\text{in}} + x_{nm}^{\text{out}})}{B_{nm} \log_2 (1 + \gamma_{nm})},
\end{aligned}
\end{equation}
where $\kappa_{n}$ denotes the maximum latency threshold for the service. This threshold is crucial for on-demand service provision, allowing the model to flexibly adapt to real-world application constraints. The computational cost of processing $x^{\text{in}}_{nm}$ input tokens (1 token $\approx$ 4 bytes = 32 bits) and generating $x^{\text{out}}_{nm}$ output tokens is $\xi_{n}\sum_{i=0}^{x^{\text{out}}_{nm}-1} \left( x^{\text{in}}_{nm} + i \right)$, where $\xi_{n}$ denotes the computational resource required per token. This quadratic growth underscores the necessity of efficient resource management. The $\ln(1/\hat{\theta}_{nm})$ is defined as the accuracy cost. As the performance gap $\hat{\theta}_{nm}$ decreases, indicating a higher accuracy demand for AIGC, the corresponding accuracy cost increases logarithmically~\cite{xu2024cached}. We adopt orthogonal frequency division multiple access (OFDMA) for data transmission to avoid interference between communication links~\cite{singhal2017resource}. $\gamma_{nm} =\frac{g_{nm}p_{n}^{t}}{\sigma_{nm}^{2}}$ denotes the signal-to-noise ratio (SNR) for the communication between MU $m$ and ASP $n$, where $g_{nm}$ is the channel gain, $p_{n}^{t}$ is the transmission power of ASP $n$,
and $\sigma$ is the additive white Gaussian noise (AWGN).

The parameters $\hat{\theta}_{nm}$, $x^{\text{in}}_{nm}$, and $x^{\text{out}}_{nm}$ in the QoE correspond to ``accuracy'', ``input token'', and ``output token'' in Fig.~\ref{flow}, respectively. MUs can adjust accuracy and input/output token count to achieve personalized AIGC services within the maximum tolerable latency set by the ASP. This flexibility enables the QoE metric to adapt to diverse application scenarios, ensuring optimized service quality tailored to both system constraints and user preferences.

\section{Problem Formulation}\label{IV}
We first define the utility functions of ASPs and MUs in Section \ref{IV-A} and \ref{IV-B}. Then, we model the interaction between MUs and ASPs as an EPEC in Section \ref{IV-C}.

\subsection{Utility of ASPs}\label{IV-A}
ASPs aim to provide AIGC services that satisfy MUs' demands while maximizing their rewards. The utility function of ASP $n$ comprises two key components: the rewards received from MUs and the costs associated with generating and transmitting AIGC. The optimization problem for ASP $n$ can be formulated as
\begin{small}
\begin{equation}
\begin{aligned}
\max\, \, &U^{asp}_{n}(\bm{f}_{n}, \bm{B}_{n})=\sum_{m=1}^M(R_{nm}\mathcal{Q}_{nm}-c_{n}^{f}f_{nm}-c_{n}^{B}B_{nm}),\\
s.t.\ \, &C1: \mathcal{Q}_{nm} \geq 0,\\
&C2: \sum_{m=1}^Mf_{nm} \leq  f_{n}^{\max}, \,\, \sum_{m=1}^MB_{nm} \leq  B_{n}^{\max},\\
\end{aligned}
\end{equation}
\end{small}where $R_{nm}$ is the reward for unit QoE value, and $\mathcal{Q}_{nm}$ denotes the quantified QoE provided by ASP $n$ to MU $m$. The decision variables $f_{nm}$ and $B_{nm}$ are the computational and communication resources allocated by ASP $n$ for MU $m$, respectively. The cost parameters $c_n^f$ and $c_n^B$ represent the unit costs of these resources. 
Constraint C1 ensures that the services provided by ASP $n$ are meaningful and viable, which is achieved by controlling the allocation of computational and communication resources\footnote{While negative utility may arise from service constraints $\mathcal{Q}_{nm} \geq 0$, it reflects rational short-term tradeoffs (e.g., maintaining market presence). In practice, non-negativity can be guaranteed by adding a baseline cost compensation $c_{n}^{f}f_{n}^{\max}+c_{n}^{B}B_{n}^{\max}$ to $U_{n}^{asp}$. This compensation reflects the maximum costs that an ASP may incur, ensuring that its utility function remains non-negative even in the worst-case scenario with minimal rewards. Since this term is a constant offset, it does not affect the optimal strategies or equilibrium analysis~\cite{Transformations} and is therefore omitted for simplicity.}. The constraint C2 ensures that the total resources of ASP $n$ remain within their respective limits, where $f_{n}^{\max}$ and $B_{n}^{\max}$ represent the maximum available computational and communication resources.
The utility function is designed to capture the trade-off between the rewards received from MUs and the costs incurred in providing AIGC services. By maximizing this utility, the ASP determines the optimal joint resource allocation, i.e., $\bm{f}_{n}$ and $\bm{B}_{n}$, which increases profitability and reduces resource consumption.

\subsection{Utility of MUs}\label{IV-B} 
Higher QoE improves user satisfaction, but offering higher rewards to ASPs increases costs. The utility function of MU $m$ captures this trade-off between service satisfaction and cost. To quantify service satisfaction, we use a gain function that models the benefits derived from AIGC services.
Specifically, we adopt a logarithmic gain function~\cite{8758205}, defined as $\mathcal{G}_{nm}=\mu_{m}ln(1\textrm{+}\sum_{n=1}^N\mathcal{Q}_{nm})$, where $\mu_{m}$ is the profit conversion coefficient.
The logarithmic form captures diminishing returns, meaning that as QoE increases, user satisfaction improves but at a decreasing rate. This reflects the realistic observation that users experience higher marginal gains at lower QoE levels, while improvements beyond a certain threshold yield smaller incremental benefits. 
Building on this gain function, the optimization problem for MU $m$ is formulated as
\begin{equation}\label{user}
\begin{aligned}
\max \, \,&U^{mu}_{m}(\bm{R}_{m})=\mu_{m}ln(1\textrm{+}\sum_{n=1}^N\mathcal{Q}_{nm})-\sum_{n=1}^NR_{nm}\mathcal{Q}_{nm} ,\\
s.t.\, \,  & R_{m}^{\min} \leq R_{nm} \leq R_{m}^{\max},  \ n \in \mathcal{N},\\
\end{aligned}
\end{equation}where $\bm{R}_{m} \triangleq \left( R_{1m}, \ldots, R_{Nm} \right)$ is the reward profile of MU $m$, with $R_{m}^{\min}$ and $R_{m}^{\max}$ denoting its minimum and maximum payable rewards, respectively. By optimizing $R_{nm}$, the MU can balance the service cost with the satisfaction derived, ensuring cost-effective access to personalized AIGC services. Given the competition for limited ASP resources, a higher reward from one MU can lead to better service, potentially reducing the QoE for others. Since MUs act selfishly and independently, their interactions naturally form a non-cooperative game, termed the multi-MU reward game $\psi$, where MUs adjust their rewards to maximize satisfaction while minimizing costs.

\begin{myDef} 
A multi-MU reward game is a tuple $\psi =\left \{ \mathcal{M}, \bm{R}, \bm{U^{mu}}\right \}$ defined by 
\begin{itemize}{
		\item {Players: The set of MUs.}
		\item { Strategies: The reward decisions $\bm{R}_{m}$ of any MU $m$.}
		\item { Utilities: The vector $\bm{U^{mu}}=\left \{ U^{mu}_{1},  \ldots,  U^{mu}_{M} \right \}$
			contains the utility functions of all the MUs defined in (\ref{user}).}}
\end{itemize}
\end{myDef} 

ASPs optimize resource allocation based on the rewards provided by MUs to maximize their benefits. Meanwhile,  MUs strategically adjust these rewards to obtain AIGC services that meet their requirements. This interdependent decision-making process creates a hierarchical game structure, where each entity optimizes its utility, leading to a complex equilibrium dynamic. Such interactions naturally follow a multi-leader, multi-follower framework, with MUs as the leaders and ASPs as the followers. In the following section, we formalize this relationship as an EPEC problem, establishing the foundation for equilibrium analysis.

\subsection{Equilibrium
Problem with Equilibrium Constraints}\label{IV-C}
We model the incentive mechanism between ASPs and MUs
as an equilibrium problem with equilibrium constraints (EPEC)~\cite{SE}, where MUs are leaders and ASPs are followers.
At the upper level, MUs determine the rewards based on
ASPs' responses and decisions of other MUs. At the lower
level, each ASP optimizes the allocation of computational and
communication resources by considering its constraints, costs,
and rewards from MUs.
This setup results in a Stackelberg equilibrium, where ASPs (followers) choose their best responses, and MUs (leaders) maximize their utilities. Our objective is to achieve a hierarchical equilibrium, characterized by a Nash equilibrium among MUs and a Stackelberg equilibrium between MUs and ASPs. The equilibria for the two levels are defined as follows.
\begin{myDef}
    Let $(\bm{f}^{*}_{n}, \bm{B}^{*}_{n})$ and $\bm{R}^{*}_{m}$ denote the optimal resource allocation of ASP $n$ and the optimal reward decision of MU $m$, respectively. Then, the points $(\bm{f}^{*}_{n}, \bm{B}^{*}_{n})$ and $\bm{R}^{*}_{m}$ are the equilibria at two levels if the following conditions hold:
\begin{equation}
\begin{aligned}
&U^{asp}_{n}\left((\bm{f}_{n}^{*},\bm{B}_{n}^{*}),\bm{R}_{m}^{*}\right)\geq U^{asp}_{n}((\bm{f}_{n},\bm{B}_{n}),\bm{R}_{m}^{*}), \\
&U^{mu}_{m}\left ( \bm{f}_{n}^{*}(\bm{R}_{m}^{*},\bm{R}_{-m}^{*}), \bm{B}_{n}^{*}(\bm{R}_{m}^{*},\bm{R}_{-m}^{*}) ,\bm{R}_{m}^{*},\bm{R}_{-m}^{*}\right )\geq \\&U^{mu}_{m}\left ( \bm{f}_{n}^{*}(\bm{R}_{m}^{*},\bm{R}_{-m}^{*}), \bm{B}_{n}^{*}(\bm{R}_{m}^{*},\bm{R}_{-m}^{*}), \bm{R}_{m},\bm{R}_{-m}^{*}\right ),\\
&\forall n \in \mathcal{N}, \ m\in \mathcal{M},
\end{aligned}
\end{equation}where $\bm{R}_{-m}$ denotes the reward profile of all MUs except MU $m$.
\end{myDef}
In summary, 
the MUs' optimization problems can be
formulated as the following EPEC problems:
    \begin{equation}\label{pro-epec}
\begin{aligned}
& \max_{\bm{R}_{m}} \, \, U^{mu}_{m}=\mu_{m}ln(1\textrm{+}\sum_{n=1}^{N}\mathcal{Q}_{nm}^{*} )- \sum_{n=1}^NR_{nm}\mathcal{Q}_{nm}^{*}, \\
& \text{s.t.} \quad \left\{ \begin{aligned}
 & R_{m}^{\min} \leq R_{nm} \leq R_{m}^{\max},  \ n \in \mathcal{N},\\
&\mathcal{Q}_{nm}^{*}=\kappa_{n}- \frac{\xi_{n} \ln(1/\hat{\theta}_{nm}) \sum_{i=0}^{x^{\text{out}}_{nm}-1} \left( x^{\text{in}}_{nm} + i \right)}{f_{nm}^{*}}\\& \,\, \, \, \, \, \,  \, \, \, \, \, \, \, \, \,\, \, \, \, \, \, \, \, \, \, \, - \frac{32 (x_{nm}^{\text{in}} + x_{nm}^{\text{out}})}{B_{nm}^{*} \log_2 (1 + \gamma_{nm})},\\
 &(\bm{f}_{n}^{*},\bm{B}_{n}^{*}) = \arg\max  U^{asp}_{n}(\bm{f}_{n}, \bm{B}_{n}), \\
 &\qquad\qquad\quad\text{subject to} \, C1, C2.
\end{aligned} \right.
\end{aligned}
\end{equation}
To solve the above EPEC, we use backward induction to address the lower-level (utility maximization for ASPs) and upper-level (non-cooperative game among MUs) problems.

\section{Equilibrium Analysis And Solutions}\label{V}
In this section, we use backward induction to investigate the existence and uniqueness of equilibria for the EPEC, i.e., the optimal resource allocation of ASPs and reward strategies of MUs.
\subsection{Lower Level: Optimal Resource Allocation for ASPs}
In the lower level of EPEC, for any reward decisions $\bm{R}$ given by MUs, ASP $n$ aims to determine its optimal computational and communication resource allocations, i.e., $\bm{f}_{n}^{*}$ and $\bm{B}_{n}^{*}$, to maximize its utility $U^{asp}_{n}(\bm{f}_{n}, \bm{B}_{n})$. In the following, we analyze and derive the unique optimal allocation.
\begin{myTheo}\label{lower_the}
 The utility maximization problem for ASP $n$ has a unique optimal solution $(\bm{f}_{n}^{*},\bm{B}_{n}^{*})$.
 \end{myTheo}
\begin{proof}
We first examine the Hessian matrix of ASP $n$'s utility $U^{asp}_{n}(\bm{f}_{n},\bm{B}_{n})$ with respect to $f_{nm}$ and $B_{nm}$. Let this Hessian matrix be denoted by 
$\bm{H}_{n}$, which can be block-diagonalized as
    \begin{equation}
        \bm{H}_{n} = \begin{bmatrix}
        \bm{H}_{n}^{f} & 0 \\
        0 & \bm{H}_{n}^{B}
        \end{bmatrix}.
    \end{equation}
     The block matrix $\bm{H}^{f}_{n}$ can be computed as the second-order partial derivative of $U^{asp}_{n}(\bm{f}_{n},\bm{B}_{n})$ with respect to $f_{nm}$, i.e.,
     \begin{small}
         \begin{equation}
 \begin{aligned}
    \bm{H}^{f}_{n} = -\mathrm{diag}\left[
         \frac{2R_{n1}\sum_{i=0}^{x^{\text{out}}_{n1}-1} \left( x^{\text{in}}_{n1} + i \right)\xi_{n} \ln(1/\hat{\theta}_{n1})}{f^{3}_{n1}},\right.&\ \\
         \left. \frac{2R_{n2}\sum_{i=0}^{x^{\text{out}}_{n2}-1} \left( x^{\text{in}}_{n2} + i \right)\xi_{n} \ln(1/\hat{\theta}_{n2})}{f^{3}_{n2}},\right.&\ \\
         \left. \cdots, \right.&\ \\
         \left. \frac{2R_{nM}\sum_{i=0}^{x^{\text{out}}_{nM}-1} \left( x^{\text{in}}_{nM} + i \right)\xi_{n} \ln(1/\hat{\theta}_{nM})}{f^{3}_{nM}}\right] &< \bm{0}.
 \end{aligned}
\end{equation}
     \end{small}
    Similarly, the block matrix $\bm{H}^{B}_{n}$ can be calculated by
    \begin{equation}
    \begin{aligned}
  \bm{H}^{B}_{n} = -\mathrm{diag}\left[
         \frac{64(x_{n1}^{\text{in}} + x_{n1}^{\text{out}})R_{n1}}{B^{3}_{n1}\log_{2}(1+\gamma_{n1})},\right.&\ \\
         \left. \frac{64(x_{n2}^{\text{in}} + x_{n2}^{\text{out}})R_{n2}}{B^{3}_{n2}\log_{2}(1+\gamma_{n2})},\right.&\ \\
         \left. \cdots, \right.&\ \\
         \left. \frac{64(x_{nM}^{\text{in}} + x_{nM}^{\text{out}})R_{nM}}{B^{3}_{nM}\log_{2}(1+\gamma_{nM})}\right] &< \bm{0}.
    \end{aligned}
    \end{equation}
    
Since $R_{nm}\geq R_{nm}^{\min}>0$, it can be shown that both $\bm{H}^{f}_{n}$ and $\bm{H}^{B}_{n}$ are diagonal matrices with strictly negative diagonal elements. Thus, the block-diagonal Hessian matrix $\bm{H}_{n}$ is negative definite, which implies that the utility function $U^{\text{asp}}_{n}(\bm{f}_{n}, \bm{B}_{n})$ is strictly concave and continuous.
Furthermore, based on the uniqueness condition established in~\cite{zhang2011game}, it follows that ASP $n$ has a unique optimal solution $(\bm{f}_{n}^{*}, \bm{B}_{n}^{*})$.
\end{proof}

\subsection{Upper Level: Optimal Reward Equilibrium among MUs}
In the upper level of the EPEC, each MU $m$ competes with other MUs and determines its reward vector $\bm{R}_{m}$. Given the ASPs' responses $(\bm{f}, \bm{B})$ and other MUs' decisions $\bm{R}_{\textrm{-}m}$,  MU $m$ selects its optimal reward $\bm{R}_m$ by maximizing its utility $U^{\text{mu}}_m(\bm{R}_m)$.

\begin{myTheo} \label{upper_the}
A unique Nash equilibrium (NE) exists in the multi-MU reward game $\psi$.
\end{myTheo}
\begin{proof}
	We define the Hessian matrix of $U^{mu}_{m}$ with respect to $\bm{R}_{m}$ as $(\bm{\Lambda}_{m}+\bm{H}_{m})$. The matrix $\bm{\Lambda}_{m}=diag \left(\frac{\partial^{2} U^{mu}_{m}}{\partial R_{1m}^{2}},\ldots,\frac{\partial^{2} U^{mu}_{m}}{\partial R_{Nm}^{2}} \right)$ and
	the second-order partial derivative matrix $\bm{H}_{m}$ is expressed by
	\begin{equation}
	\bm{H}_{m}=\begin{bmatrix}
	0 &  \frac{\partial^{2} U^{mu}_{m}}{\partial R_{1m}\partial R_{2m}}& \cdots  & \frac{\partial^{2} U^{mu}_{m}}{\partial R_{1m}\partial R_{Nm}}\\ 
	\frac{\partial^{2} U^{mu}_{m}}{\partial R_{2m}\partial R_{1m}} & 0 & \cdots  & \frac{\partial^{2} U^{mu}_{m}}{\partial R_{2m}\partial R_{Nm}}\\ 
	\vdots & \vdots  &  \ddots & \vdots \\ 
	\frac{\partial^{2} U^{mu}_{m}}{\partial R_{Nm}\partial R_{1m}}&  \frac{\partial^{2} U^{mu}_{m}}{\partial R_{Nm}\partial R_{2m}}&\cdots   & 0
	\end{bmatrix},
	\end{equation}	where 
	\begin{equation}
\begin{aligned}
\frac{\partial^{2} U^{mu}_{m}}{\partial R_{nm}^{2}}=&\ \mu_{m}\frac{\mathcal{Q}_{nm}''(1+\sum_{n=1}^N\mathcal{Q}_{nm})-(\mathcal{Q}_{nm}')^{2}}{(1+\sum_{n=1}^N\mathcal{Q}_{nm})^{2}}-2\mathcal{Q}_{nm}'\\&-R_{nm}\mathcal{Q}_{nm}''<0, \ n \in \mathcal{N},
\end{aligned}
\end{equation}
\begin{equation}
\begin{aligned}
    \frac{\partial^{2} U^{mu}_{m}}{\partial R_{nm}\partial R_{n'm}}=&-\mu_{m}\frac{\mathcal{Q}_{nm}'\mathcal{Q}_{n'm}'}{(1+\sum_{n=1}^N\mathcal{Q}_{nm})^{2}}<0,  \\ &\  \forall n \in \mathcal{N}, \ n\neq n'.
\end{aligned}
\end{equation}
	The proof of $\frac{\partial^{2} U^{mu}_{m}}{\partial R_{nm}^{2}}<0$ and $\frac{\partial^{2} U^{mu}_{m}}{\partial R_{nm}\partial R_{n'm}}<0$ are omitted due to space limits. We randomly choose a vector $\bm{h}\in \mathbb{R}^{N\times 1}$ with elements not all $0$. Then, we have $\bm{h}^{T}(\bm{\Lambda}_{m}+\bm{H_{m}})\bm{h}=\sum_{n=1}^N (h^n)^2 (\frac{\mu_m \mathcal{Q}_{nm}''}{1+\sum_{n=1}^N\mathcal{Q}_{nm}} - 2\mathcal{Q}_{nm}' - R_{nm}\mathcal{Q}_{nm}'' )-\frac{\mu_m}{(1+\sum_{n=1}^N\mathcal{Q}_{nm})^2} ( \sum_{n=1}^N h^n \mathcal{Q}_{nm}' )^2 <0$, indicating that the utility function $U_{m}^{mu}$ is strictly concave. According to~\cite{NASH}, there exists a unique Nash equilibrium in the multi-MU reward game $\psi$.   
\end{proof}
We verify the effectiveness of the incentive mechanism by proving the optimality and uniqueness of MUs' reward decisions and ASPs' resource allocation through~\textbf{Theorem~\ref{lower_the}} and \textbf{Theorem~\ref{upper_the}}.

\subsection{The Process of the QoE-driven Incentive Mechanism}
The QoE-driven incentive process dynamically adjusts the reward strategies for MUs based on the ASPs' responses. This approach effectively aligns the interests of both ASPs and MUs, facilitating the delivery of personalized AIGC services. The detailed process of this incentive mechanism is outlined in Algorithm \ref{alg:price_adjustment}.

\subsubsection{The Solution Process}
The algorithm initializes each MU 
$m$ with a reward vector $\bm{R}_m^{(0)}$ and a step size $\Delta$. At each iteration $t$, given the current rewards, each ASP $n \in \mathcal{N}$ determines its optimal computational and communication resource allocation $(\bm{f}_{n}^*, \bm{B}_{n}^*)$. Each MU $m \in \mathcal{M}$ updates its reward $R_{nm}^{(t)}$ offered to ASP $n$ by adjusting it by $\Delta_t$, while keeping the rewards $\bm{R}_{-n,m}^{(t)}$ fixed for all other ASPs. Specifically, if  increasing the reward by $\Delta_t$ maximizes the utility $U_m^{mu}(\bm{R}_m)$, the reward is updated to $R_{nm}^{(t)} + \Delta_t$; if decreasing by $\Delta_t$ yields the highest utility, it is set to
$R_{nm}^{(t)} - \Delta_t$; otherwise, the reward remains unchanged. 
The process continues
until no MU can improve its utility by unilaterally adjusting its reward, which
is verified when the total change in utilities between consecutive iterations falls
below a predefined threshold $\epsilon$.
At this point, the algorithm converges to obtain the final reward strategy $\bm{R}^*$ and the corresponding resource allocation $(\bm{f}^*, \bm{B}^*)$.

\begin{algorithm}[t]
\caption{Dual-Perturbation Reward Optimization}
\label{alg:price_adjustment}
\begin{algorithmic}[1]
\Require Personalized demands of MUs and system parameters of ASPs
\Ensure Optimal reward $\bm{R}^*$ and resource allocation $(\bm{f}^*,\bm{B}^*)$
\State \textbf{Initialize:} Rewards $\bm{R}_{m}^{(0)}$ for MU $m$, step size $\Delta$, and convergence threshold $\epsilon$.
\For{$t = 1, 2, \dots, T$}  
    \For{each ASP $n \in \mathcal{N}$}
        \State Based on the rewards provided by all MUs, each \\\hspace{3em}ASP determines its optimal computational and \\\hspace{3em}communication resource allocation $(\bm{f}_{n}^*, \bm{B}_{n}^*)$;
    \EndFor
    \For{each MU $m \in \mathcal{M}$}
        \State Store current rewards: $\bm{R}_{m}^{(t)} = \bm{R}_{m}$;  
        \State Each MU tries to increase and decrease its rewards \\\hspace{3em}with the step size \( \Delta_t \), and calculates its own utility \\\hspace{3em}based on the ASPs' optimal strategies;
        \For{each ASP $n$}
            \If{
                $U_{m}^{mu}(R^{(t)}_{nm}, \bm{R}^{(t)}_{-n,m}) \leq U_{m}^{mu}(R^{(t)}_{nm} + \Delta_t, $\\ \hspace{4.5em}$\bm{R}^{(t)}_{-n,m})$
                \text{and} $U_{m}^{mu}(R^{(t)}_{nm} - \Delta_t, \bm{R}^{(t)}_{-n,m}) \leq $\\\hspace{4.5em}$U_{m}^{mu}(R^{(t)}_{nm} + \Delta_t, \bm{R}^{(t)}_{-n,m})$
            }
                \State $ R_{nm} = \min\{ R_{nm}^{(t)} + \Delta_t, R_{m}^{\max} \}$;  
            \ElsIf{
                $U_{m}^{mu}(R^{(t)}_{nm}, \bm{R}^{(t)}_{-n,m}) \leq U_{m}^{mu}(R^{(t)}_{nm} - $\\ \hspace{4.5em}$\Delta_t,\bm{R}^{(t)}_{-n,m})$
                \text{and} $U_{m}^{mu}(R^{(t)}_{nm} + \Delta_t, \bm{R}^{(t)}_{-n,m}) \leq $\\ \hspace{4.5em}$U_{m}^{mu}(R^{(t)}_{nm} - \Delta_t, \bm{R}^{(t)}_{-n,m})$
            }
                \State $R_{nm} = \max\{R_{m}^{\min}, R_{nm}^{(t)} - \Delta_t\}$;  
            \Else
                \State \( R_{nm} = R_{nm}^{(t)} \)  
            \EndIf
        \EndFor
    \EndFor
    \If{ \( \sum_{m=1}^{M} | U_{m}^{mu}(\bm{R}_{m}^{(t)}) - U_{m}^{mu}(\bm{R}_{m}^{(t-1)}) | \leq \epsilon \) }
        \State \( \mathbf{R}^* = \mathbf{R} \) 
        \State \( (\mathbf{f}^*, \mathbf{B}^*) = \texttt{ASP\_response}(\mathbf{R}^*) \)
        \State \textbf{break}  
    \EndIf
\EndFor
\end{algorithmic}
\end{algorithm}

\subsubsection{Execution Flow}
 At the beginning of each update cycle, MUs report their personalized demands, and ASPs provide current resource and network conditions. Based on this information, the central decision agent runs Algorithm~1 to determine optimal rewards and resource allocations. This framework integrates the real-time sensing capability of distributed nodes with centralized control, ensuring both efficiency and consistency in decision-making.
In contrast, fully distributed methods require frequent strategy exchanges among nodes during each iteration, resulting in significant communication overhead. Similarly, gradient-based distributed optimization techniques rely on each MU to accurately estimate or jointly compute the partial derivatives of its utility function, which introduces additional computational and coordination complexity.
Algorithm 1 leverages a zeroth-order perturbation mechanism~\cite{7728084} to streamline this process. In each iteration, the central decision agent perturbs each MU's reward slightly in both directions and observes the resulting utility changes. Based on this directional feedback, it simulates how each MU would react and updates the rewards accordingly. This
approach eliminates the need for gradient computation and
extensive message passing, thereby reducing overall system
overhead.

\subsubsection{Convergence Analysis}
The convergence complexity of Algorithm 1 depends on the step-size strategy~\cite{6021434}, as detailed in Appendix A. With a diminishing step size $\Delta_{t}=\frac{u}{t}$, Algorithm 1 converges to an $\epsilon$-Nash equilibrium at a rate of $T_\epsilon = \mathcal{O}(\exp(\frac{\ddot{\kappa}}{\epsilon}))$, where $\ddot{\kappa}=  (\frac{N}{u}\sum_{m=1}^M(R_{m}^{\max})^2 +  MN\frac{u\pi^2}{6})$. The analysis indicates that finding the NE in multi-MU games is typically challenging, even with a small number of MUs~\cite{RePEc}.
For a fixed number of MUs and ASPs, the convergence speed is primarily influenced by the initial reward vector $\mathbf{R}_m^{(0)}$  for each MU $m \in \mathcal{M}$ and the step size parameter $u$. If $\|\mathbf{R}_m^{(0)} - \mathbf{R}_m^*\|_2 \approx u$, rapid convergence to the equilibrium is achieved; however, suboptimal choices of $\mathbf{R}_m^{(0)}$ or $u$ may significantly slow convergence.  
Due to the determinism of the update rule, given an initialization and fixed parameters, each iteration $t$ produces a unique reward vector $\mathbf{R}_m^{(t)}$, and thus the (initialization-dependent) iterative trajectory is unique and reproducible. In contrast,
using a constant step size $\Delta_{t}=\hat{\Delta}$, the algorithm  approaches a neighborhood of the NE at a rate of $\mathcal{O}(\frac{{\kappa}'}{\epsilon'})$, where 
   $\epsilon' =  \epsilon - \frac{\hat{\Delta} MN}{2}$ and $\kappa' = \frac{N \sum_{m=1}^M (R_m^{\max})^2}{\hat{\Delta}}$.
This setting achieves faster, polynomial convergence in $1/\epsilon'$, but only ensures proximity
to the equilibrium within a fixed neighborhood. In practice, diminishing step
sizes are preferred when high solution accuracy is critical and longer convergence times are acceptable, whereas constant step sizes are more suitable when
a slight accuracy loss is tolerable (e.g., due to quantized decision variables) and
rapid convergence is prioritized.

\begin{figure*}[t]
	\centering
	\subfigure[Reward decision and QoE.] {
		\label{QoE_theta}
		\includegraphics[scale=0.38]{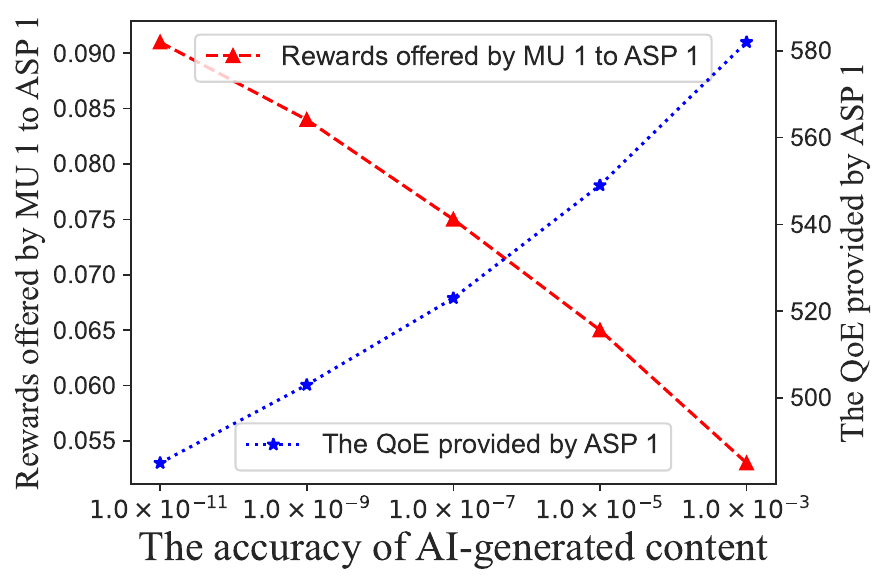}
	}
	\subfigure[Resource allocation.] {
		\label{resource_theta}
		\includegraphics[scale=0.38]{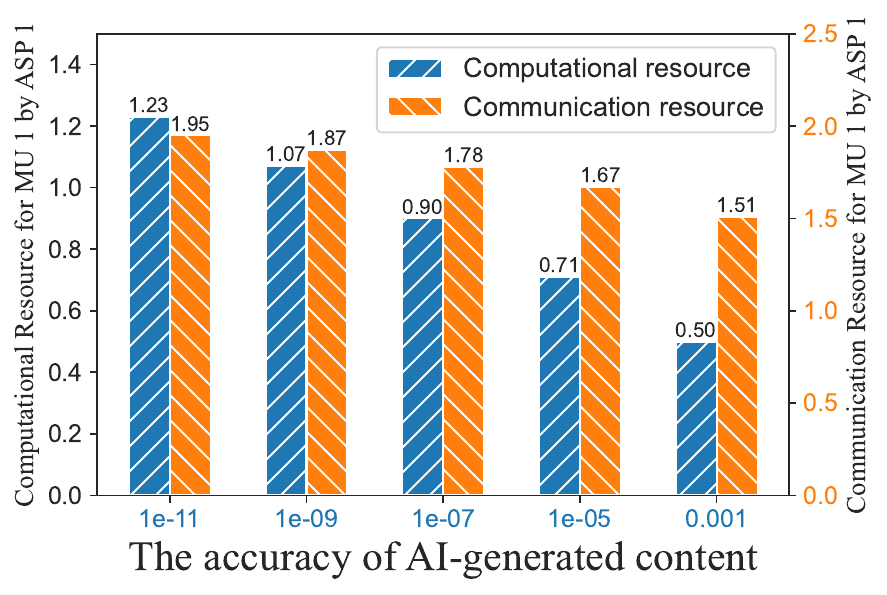}
	}
	\subfigure[Utilities of ASP 1 and MU 1.] {
		\label{ASPs_theta}
		\includegraphics[scale=0.38]{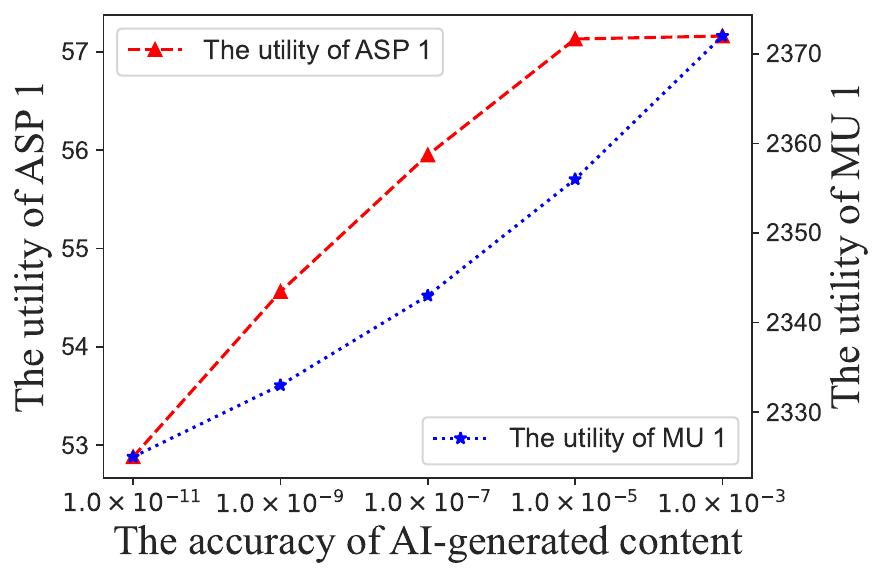}
	}
	\subfigure[Reward decision and QoE.] {
		\label{QoE_x}
		\includegraphics[scale=0.38]{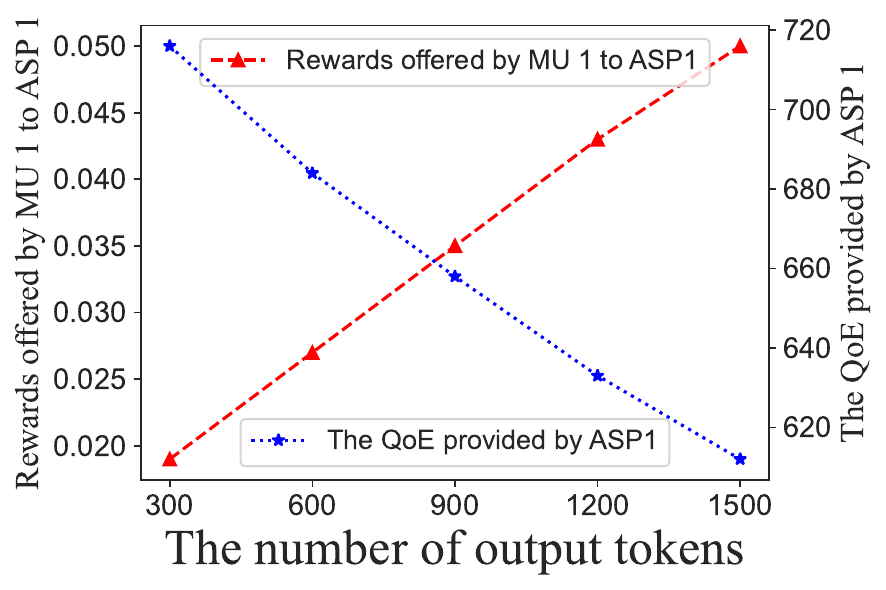}
	}
	\subfigure[Resource allocation.] {
		\label{resource_x}
		\includegraphics[scale=0.38]{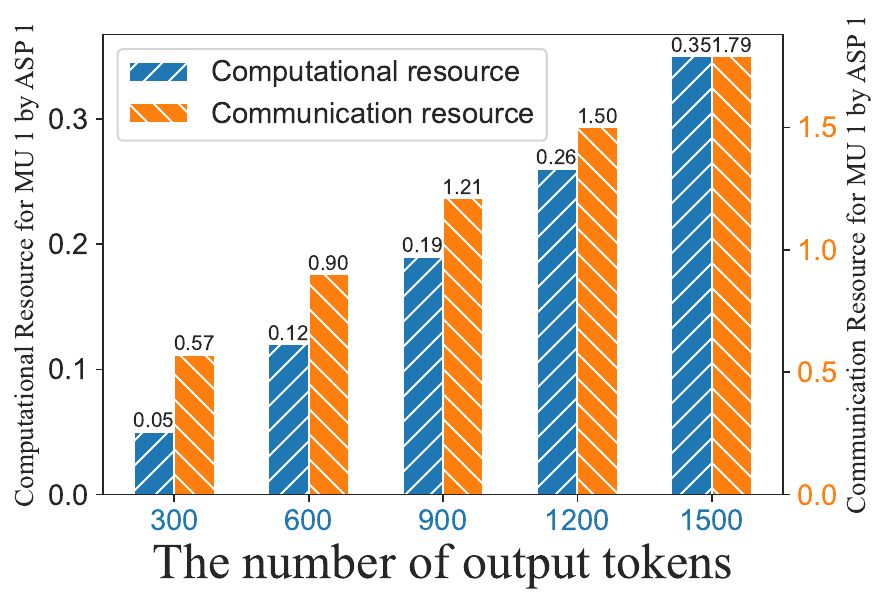}
	}
	\subfigure[Utilities of ASP 1 and MU 1.] {
		\label{ASPs_x}
		\includegraphics[scale=0.38]{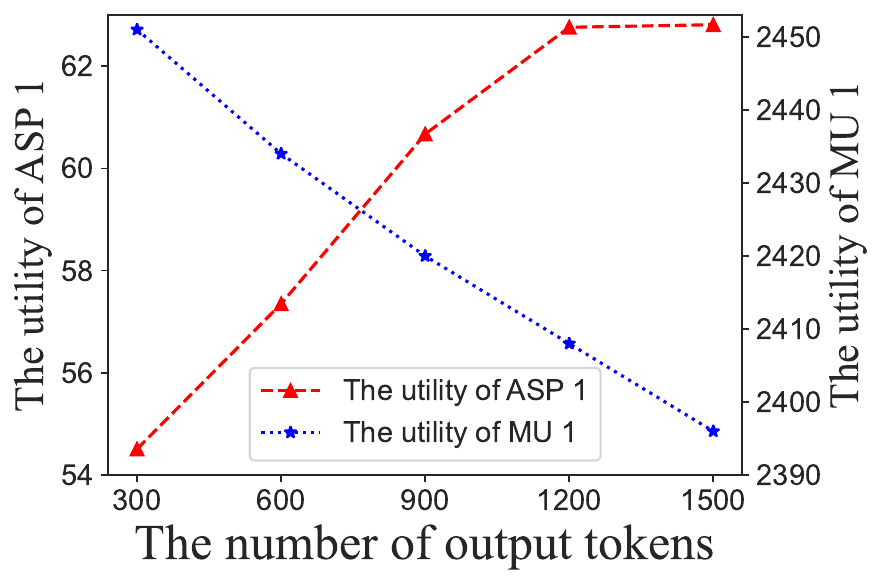}
	}
	\subfigure[Reward decision and QoE.] {
		\label{QoE_k}
		\includegraphics[scale=0.38]{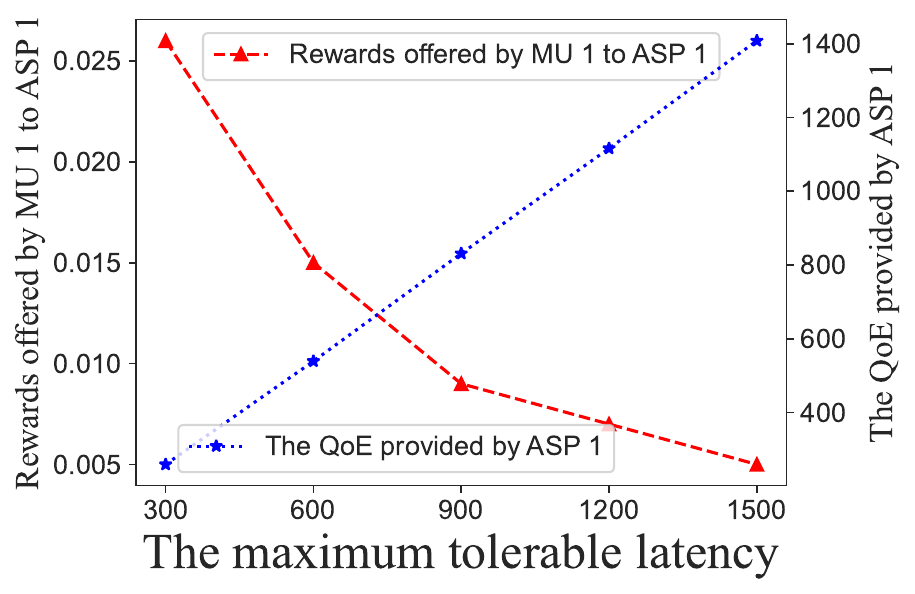}
	}
	\subfigure[Resource allocation.] {
		\label{resource_k}
		\includegraphics[scale=0.37]{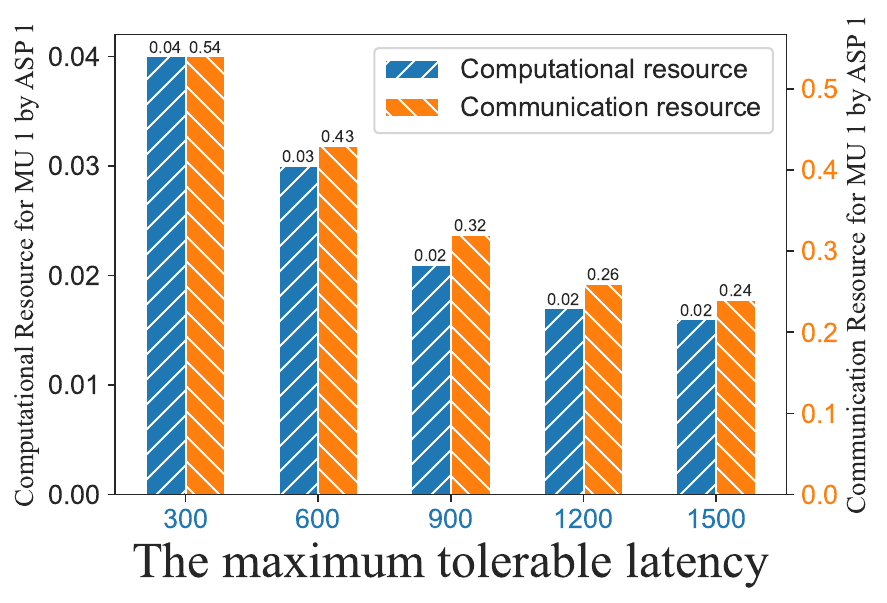}
	}
	\subfigure[Utilities of ASP 1 and MU 1.] {
		\label{ASPs_k}
		\includegraphics[scale=0.38]{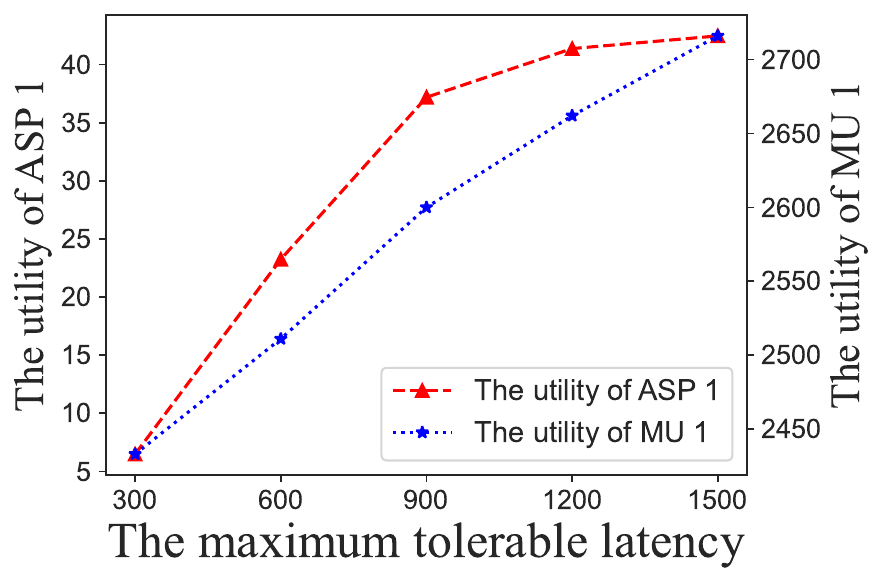}
	}
	
	\caption{The impact of three parameters on the decisions of MU 1 and ASP 1.}
	\label{combined_figures}
\end{figure*}

\begin{figure}[t]
	\centering
	\subfigure[The average utility versus number of MUs.] {
		\label{fig:MUvary_utilities}
		\includegraphics[width=1\linewidth]{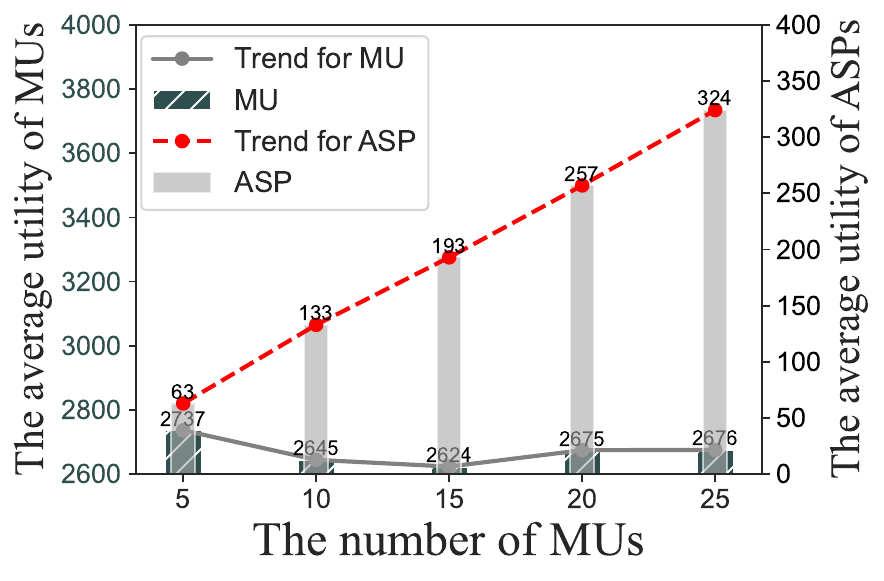}
	}
	\subfigure[The average  utility versus number of ASPs.] {
		\label{fig:ASPvary_utilities}
		\includegraphics[width=1\linewidth]{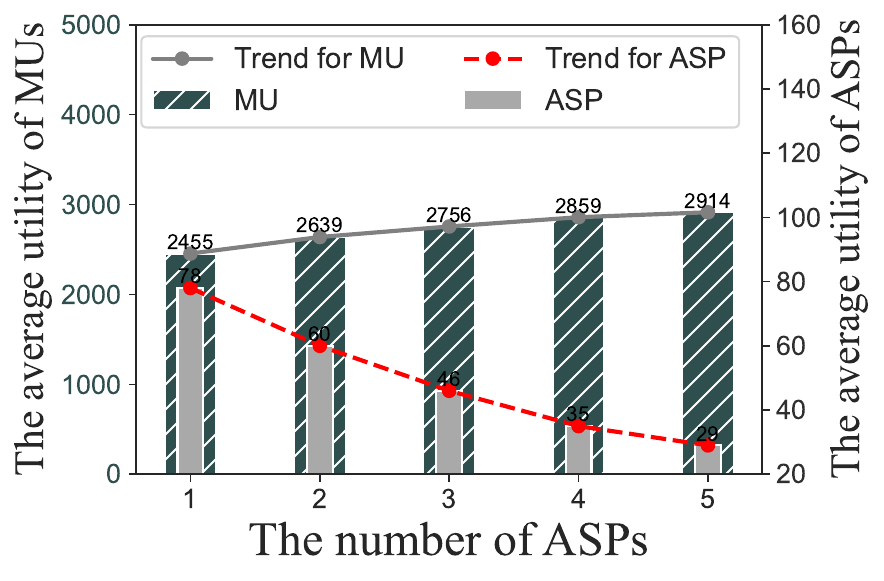}
	}
	\caption{Impacts of the number of MUs/ASPs on utilities.}
\end{figure}

\section{Simulation Results}\label{exp}
 In this section, we first show the impact of different demands on incentives through QoE quantitative analysis and a case study. Then, we demonstrate the effectiveness and advantages of the proposed incentive through comparative experiments with existing schemes.

\subsection{Analysis of Critical Factors}
To investigate the incentive process, we adopt a quantitative approach to examine how varying the personalized demands of MU 1, while keeping those of other MUs fixed, influences the trading outcomes.
We precisely manipulate three critical parameters associated with $\mathcal{Q}_{11}$: the accuracy $\hat{\theta}_{11}$ of AIGC\footnote{The ASP derives the value of $K$ based on the MUs' accuracy requirements through $\hat{\theta}_{nm}= \beta_{nm}(\frac{\eta}{1-\eta})^{K}$, thereby guiding the inference service to meet the demands of MUs.}, with values $[10^{-11}, 10^{-9}, 10^{-7},10^{-5},10^{-3}]$ depending on $K=[10,8,6,4,2]$; the number of output tokens $x_{11}$, varying across $[300, 600, 900, 1200, 1500]$; and the maximum tolerable latency $\kappa_{11}$, ranging from $[300,600,900,1200,1500]$ms. Furthermore, we analyze the average utility of  MUs and ASPs by adjusting the number of MUs and ASPs. By systematically varying these parameters, we aim to investigate their effects on the decision-making process comprehensively. The main experimental parameters and values are summarized in Table~\ref{tab:experiment_settings}.

\subsubsection{Accuracy of AIGC $\hat{\theta}$}
As the value of $\hat{\theta}_{11}$ increases (indicating lower accuracy requirements), the reward $R_{11}$ per unit of $\mathcal{Q}_{11}$ offered by MU 1 to ASP 1 decreases (Fig.~\ref{QoE_theta}). This reduction in $R_{11}$ stems from MU 1's rational response to relaxed accuracy constraints, which lowers the marginal value of additional computational efforts (Fig.~\ref{resource_theta}), thus enabling MU 1 to enjoy the service at a lower cost. However, the QoE shows an upward trend as $\hat{\theta}_{11}$ increases. This upward trend in QoE can be attributed to its definition: a higher QoE value indicates the ability to meet the requirements of MUs more quickly. Notably, as  accuracy requirements become easier to satisfy, the QoE value rises accordingly. Furthermore, MU 1 experiences an increase in utility (Fig.~\ref{ASPs_theta}). This is because the relaxed accuracy requirements decrease the total incentive rewards $R_{11}\mathcal{Q}_{11}$ for ASP 1, allowing the service to be obtained at a lower cost. Meanwhile, ASP 1 adapts its resource allocation to optimize its utility, ensuring that reduced total incentive rewards do not result in a loss of utility.
These findings show that relaxing AIGC accuracy requirements reduces computational demand and incentive costs while improving QoE, thereby validating the dynamic trade-off between quality and cost.

\begin{table}[t]
\centering
\caption{Experimental Parameter Settings.}
\label{tab:experiment_settings}
\begin{tabular}{lc}
\toprule
\textbf{Parameters} & \textbf{Values} \\
\midrule
Number of ASPs            & [1,2,3,4,5]                       \\
Number of MUs           &[5,10,15,20,25] \\
$K$ (CoT Examples)        & [2, 4, 6, 8, 10]        \\
Input/Output Token ($x^{\text{in}}$/$x^{\text{out}}$)           & 100--2000 tokens        \\
Maximum Computational Resource ($f^{\max}$) & 5-30 TFLOPS \\
Maximum Bandwidth ($B^{\max}$) & 100-500 MHz      \\
Maximum Tolerable Latency ($\kappa$)   & 300-1500 ms       \\
SNR        & 10-30 dB                 \\
\bottomrule
\end{tabular}
\end{table}

\subsubsection{Number of Output Tokens $x^{\text{out}}$}
 As the required number of output tokens increases,  the reward $R_{11}$ per unit of $\mathcal{Q}_{11}$ provided by MU 1 to ASP 1 tends to rise (Fig.~\ref{QoE_x}). This is because each additional output token requires processing more tokens due to the autoregressive nature of the model, where each output depends on all previous inputs. To meet MU 1's increasing demands, more computational and communication resources are allocated to MU 1 (Fig.~\ref{resource_x}).
 Furthermore, MU 1 experiences a decrease in utility, whereas ASP 1's utility increases with the rise in output tokens (Fig.~\ref{ASPs_x}). This trend indicates that as the demand for $x^{\text{out}}$ grows, MU 1 compensates ASP 1 with higher total incentive rewards, leading to increased costs and reduced utility. However, ASP 1 strategically optimizes its resource allocation to efficiently meet the personalized demands of MU 1, maximizing its utility within the constraints of limited resources and latency thresholds. Ultimately, these findings underscore that higher token demand increases resource consumption, making it the primary driver of cost growth while enhancing the profitability of ASPs.

\subsubsection{Maximum Tolerable Latency $\kappa$}
As the value of $\kappa$ increases,  the reward $R_{11}$ per unit of $\mathcal{Q}_{11}$ provided by MU 1 to ASP 1 gradually decreases (Fig.~\ref{QoE_k}). This occurs because the increase in $\kappa$ makes MU 1's demand less urgent, allowing ASP 1 to meet MU 1's requirements with fewer resources (Fig.~\ref{resource_k}), thereby lowering the cost per unit of QoE. 
Moreover, both the utilities of MU 1 and ASP 1 show an upward trend (Fig.~\ref{ASPs_k}).  For MU 1, lower service requirements naturally result in cost savings and higher utility. For ASP 1, a larger $\kappa$ facilitates more cost-efficient service provision, thereby enhancing its utility.
Overall, the increase in $\kappa$ allows ASPs to reduce resource consumption and increase benefits. However, to ensure fairness and consistency across tasks of the same type, the $\kappa$ value should be uniformly set by ASPs. A reasonable relaxation of latency constraints can lower costs while ensuring service quality, ultimately yielding mutual benefits for both MUs and ASPs.

\subsubsection{Number of MUs and ASPs}
The increasing trends observed in Fig.~\ref{fig:MUvary_utilities} for ASPs are attributed to continued market expansion. With the increase in MUs, the demand for their services expands, thereby driving the profitability of ASPs. This growth reflects the positive correlation between the number of MUs and the average utility of ASPs.  In contrast, the stability of the average MU utility (fluctuation range $<$ $5\%$) stems from the emergence of a Nash equilibrium in the non-cooperative game. When new MUs enter the system, existing MUs strategically adjust their reward strategies to respond to competition and maximize their utilities.
As the number of ASPs increases, the average utility of MUs gradually rises, while the average utility of ASPs decreases (Fig.~\ref{fig:ASPvary_utilities}).
This is because when new ASPs are added, MUs reduce the total incentive rewards for each ASP to lower costs, thereby increasing their utilities.
However, this reward reduction directly affects the ASP's revenue, leading to a decline in its average utility. Furthermore, MUs' benefits are based on the QoE evaluation of all tasks, making MUs more concerned with the overall benefits of all tasks rather than an individual task. 
The above findings indicate that, as the market size expands, the average utility of MUs remains relatively stable, while the benefits for ASPs are influenced by the number of MUs and ASPs.

\begin{table}[t]
\caption{Personalized Demands of MUs.}
\centering
\label{table1}
\begin{threeparttable}
\begin{tabular}{cclllc}\\
\toprule
\multicolumn{1}{l}{} & \multicolumn{4}{c}{AIGC Service Provider 1}   & \multicolumn{1}{l}{AIGC Service Provider 2} \\ \midrule 
MU 1                 & \multicolumn{4}{c}{$(\theta= 1e^{-7}, x^{\text{out}}=200)$} & $(\theta=1e^{-5},x^{\text{out}}=1400)$                    \\
MU 2                 & \multicolumn{4}{c}{$(\theta=1e^{-9}, x^{\text{out}}=500)$} & $(\theta= 1e^{-7},x^{\text{out}}=1200)$                    \\
MU 3                 & \multicolumn{4}{c}{$(\theta=1e^{-8}, x^{\text{out}}=800)$} & $(\theta= 1e^{-8},x^{\text{out}}=1000)$                    \\ \bottomrule
\end{tabular}
\begin{tablenotes}
\item {Note: The maximum tolerable latency of ASP 1 and ASP 2 are $500$ms and $1000$ms, respectively.}
\end{tablenotes}
\end{threeparttable}
\end{table}
\begin{table*}[t]
\caption{Results of the QoE-Driven Incentive Mechanism: \\ Reward Decisions for MUs and Resource Allocation Decisions for ASPs.}
 \centering
  \label{table2}
\begin{threeparttable}  
\begin{tabular}{ccccccccc}
\toprule
    & \multicolumn{4}{c}{AIGC Service Provider 1}                                                                                                         & \multicolumn{4}{c}{AIGC Service Provider 2}                                                                                                         \\ \cmidrule(lr){2-5} \cmidrule(lr){6-9}
  & \begin{tabular}[c]{@{}c@{}}computational\\  resource\end{tabular} & \begin{tabular}[c]{@{}c@{}}communication\\ resource\end{tabular} & rewards & QoE  & \begin{tabular}[c]{@{}c@{}}computational\\  resource\end{tabular} & \begin{tabular}[c]{@{}c@{}}communication\\ resource\end{tabular} & rewards & QoE  \\ \midrule
MU 1 & 0.4                                                              & 1.58                                                             & 0.06    & 353 & 0.62                                                             & 1.73                                                             & 0.056   & 730 \\
MU 2 & 0.93                                                             & 2.13                                                             & 0.1     & 298 & 1.23                                                             & 2.71                                                             & 0.086   & 643 \\
MU 3 & 1.44                                                             & 2.82                                                             & 0.126   & 255 & 0.99                                                             & 2.26                                                             & 0.076   & 674 \\ \bottomrule
\end{tabular}
\begin{tablenotes}
\item 
{Note: The units of computational and communication resources are TFLOPs and MHz, respectively.}
\end{tablenotes}
\end{threeparttable}
\end{table*}

\begin{figure*}[t]
	\centering
	\subfigure[Computational overhead.] {
		\label{utif_MU}
		\includegraphics[scale=0.42]{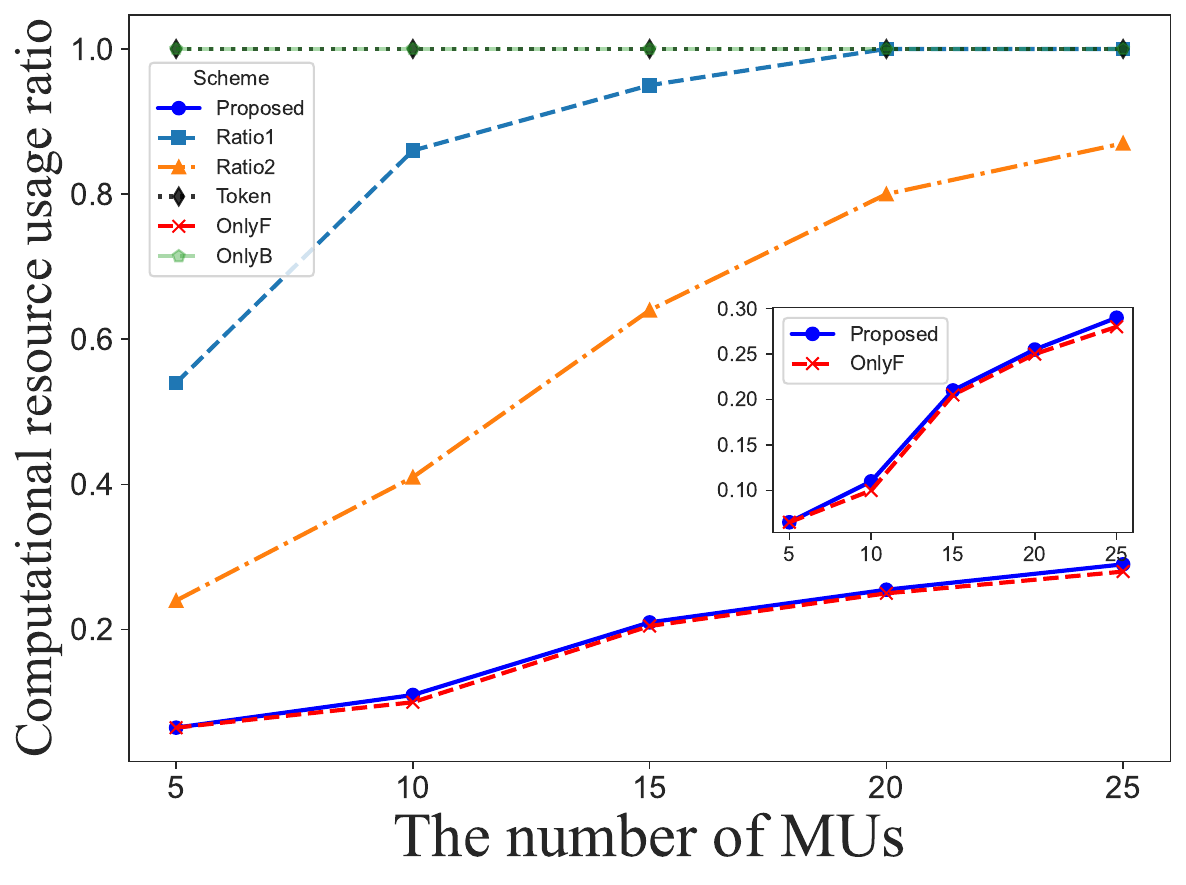}
	}
	\subfigure[Communication overhead.] {
		\label{utiB_MU}
		\includegraphics[scale=0.42]{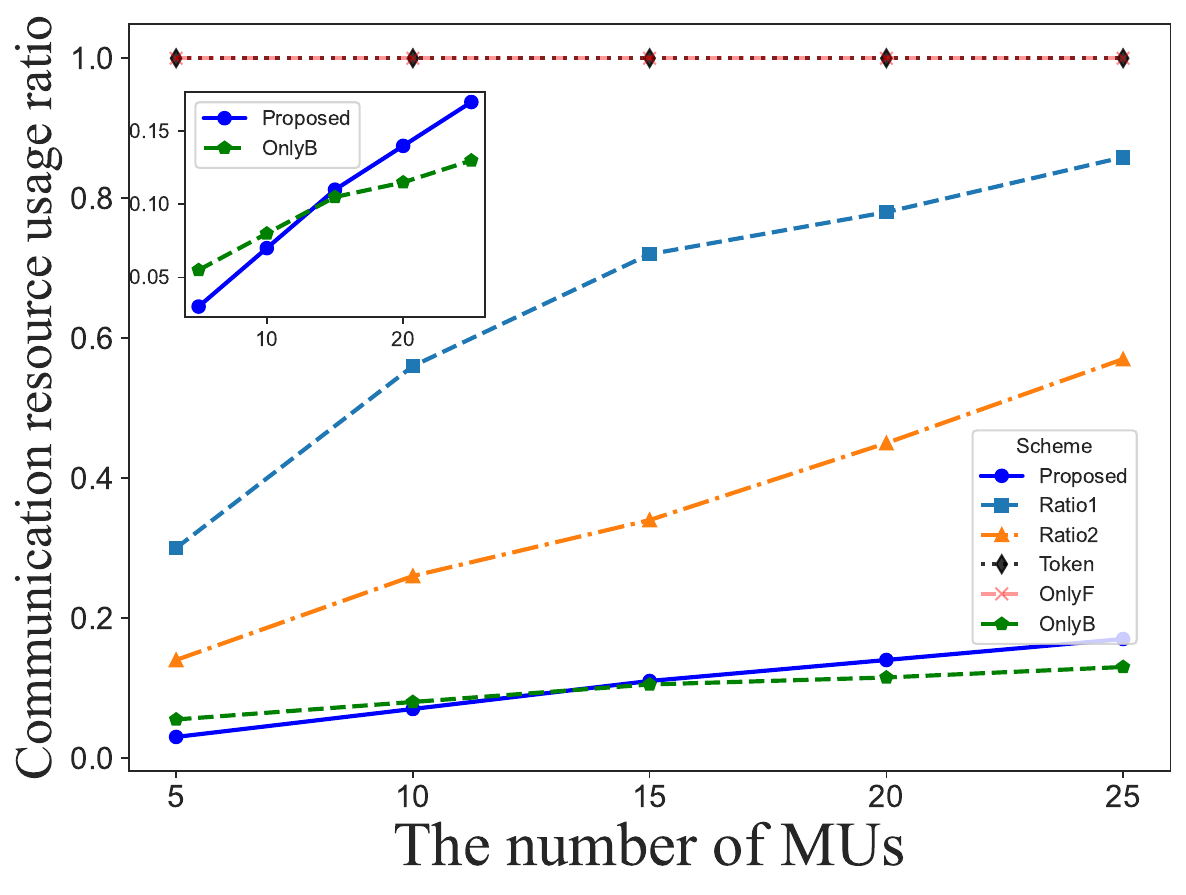}
	}
    	\subfigure[The average cost of MUs.] {
		\label{MUcost_MU}
		\includegraphics[scale=0.42]{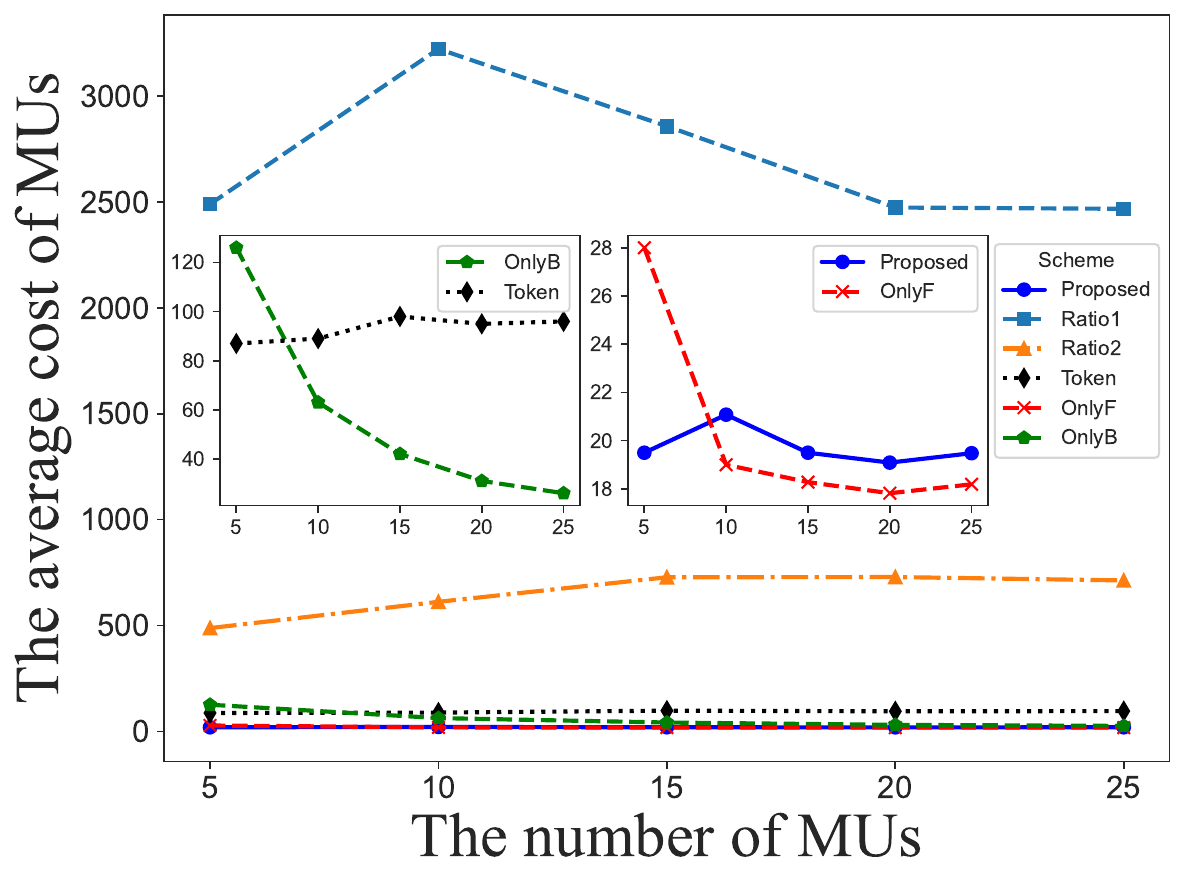}
	}
    \subfigure[The average cost of ASPs.] {
		\label{ASPcost_MU}
		\includegraphics[scale=0.42]{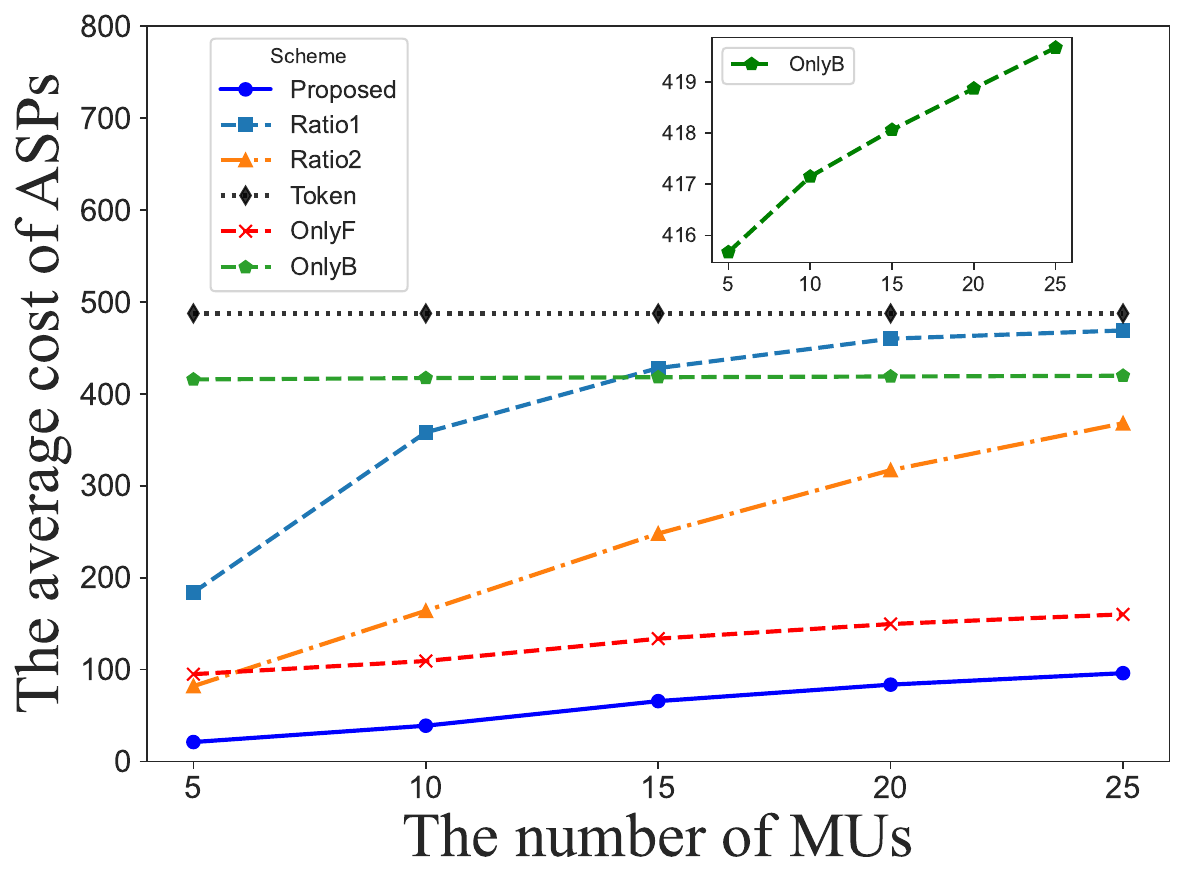}
	}
	\caption{The impact of the number of MUs on different schemes.}
	\label{compare_MU}
\end{figure*}

\begin{figure*}[t]
	\centering
	\subfigure[Computational overhead.] {
		\label{utif_ASP}
		\includegraphics[scale=0.42]{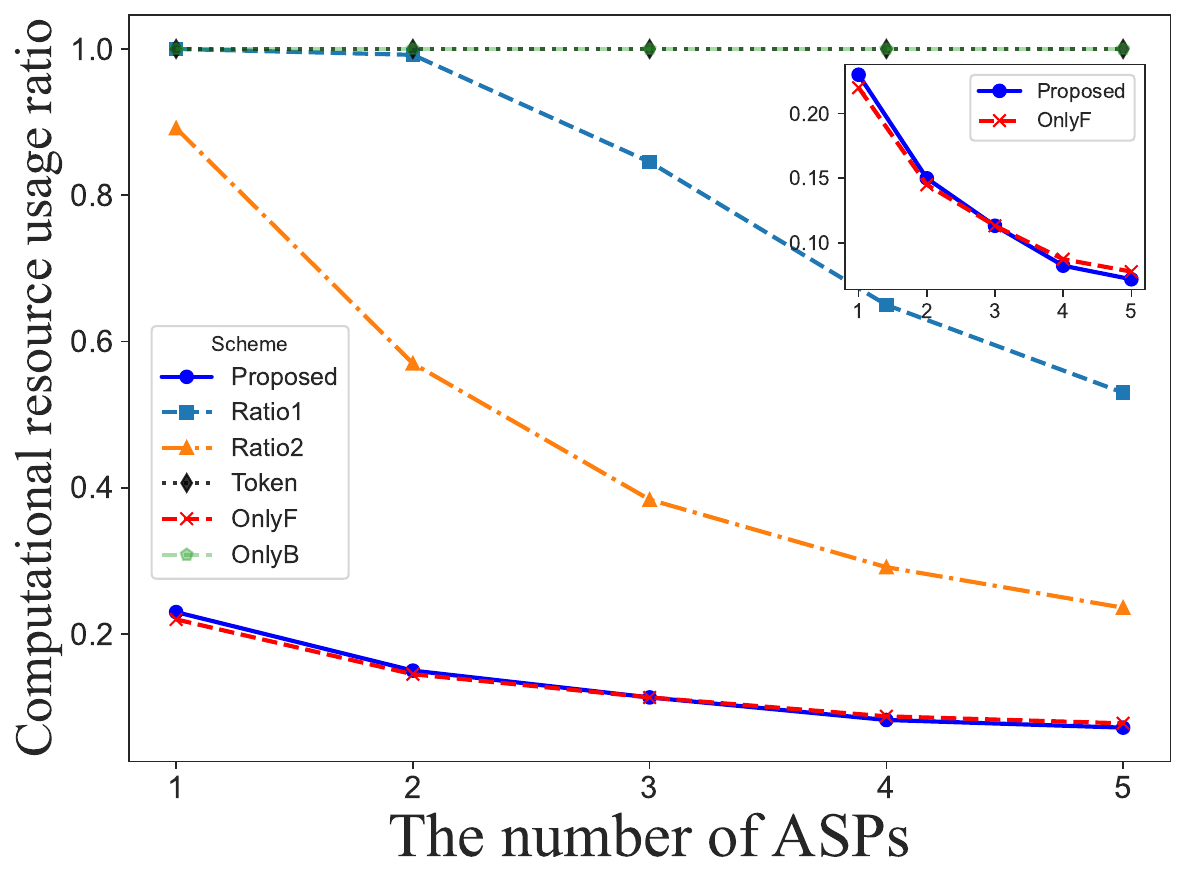}
	}
	\subfigure[Communication overhead.] {
		\label{utiB_ASP}
		\includegraphics[scale=0.42]{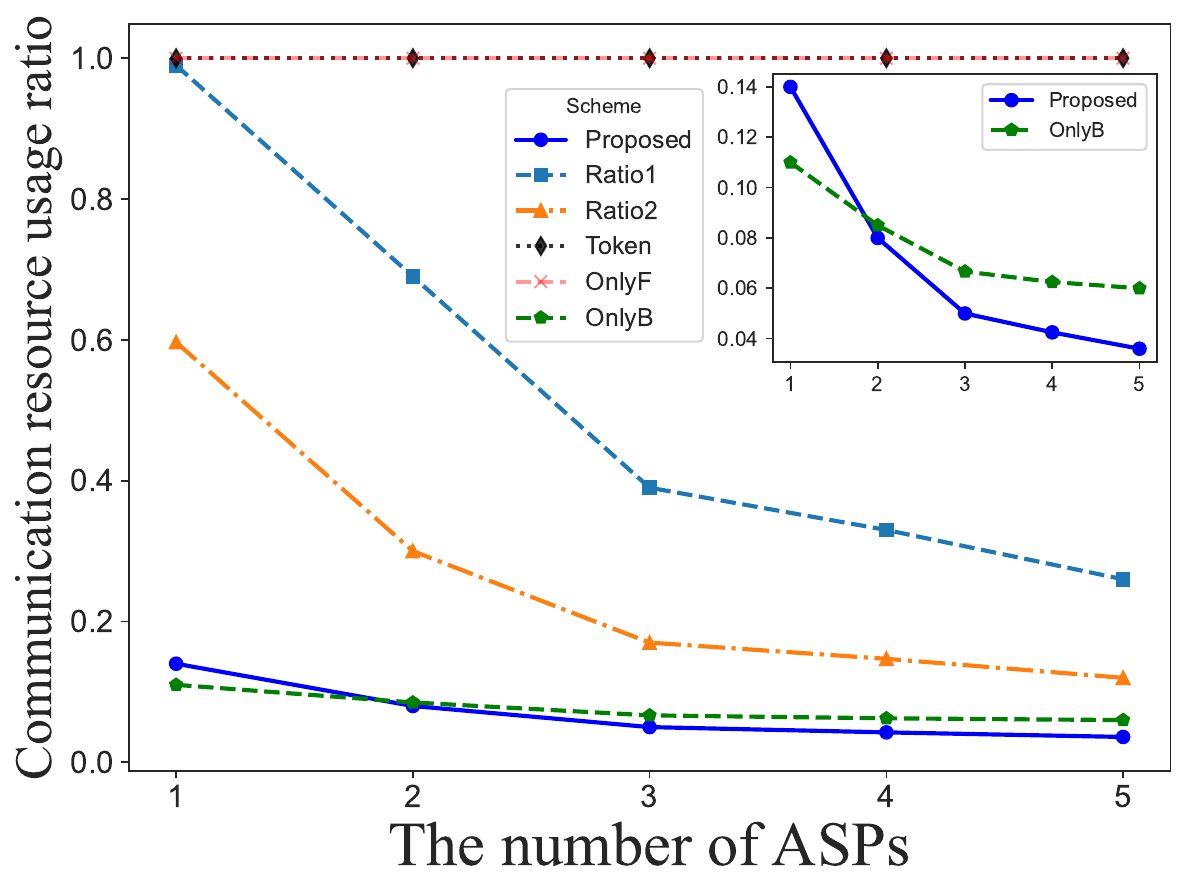}
	}
    	\subfigure[The average cost of MUs.] {
		\label{MUcost_ASP}
		\includegraphics[scale=0.42]{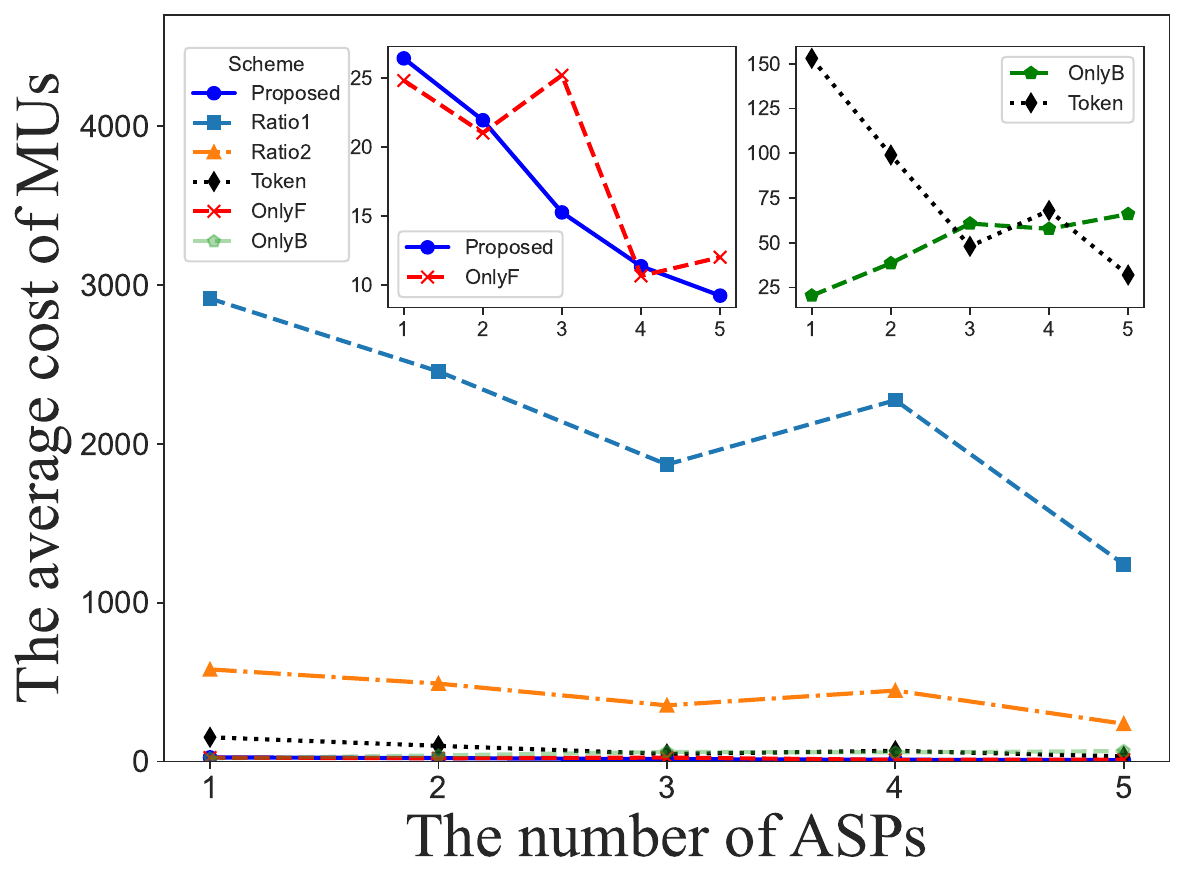}}
\subfigure[The average cost of ASPs.] {
		\label{ASPcost_ASP}
		\includegraphics[scale=0.42]{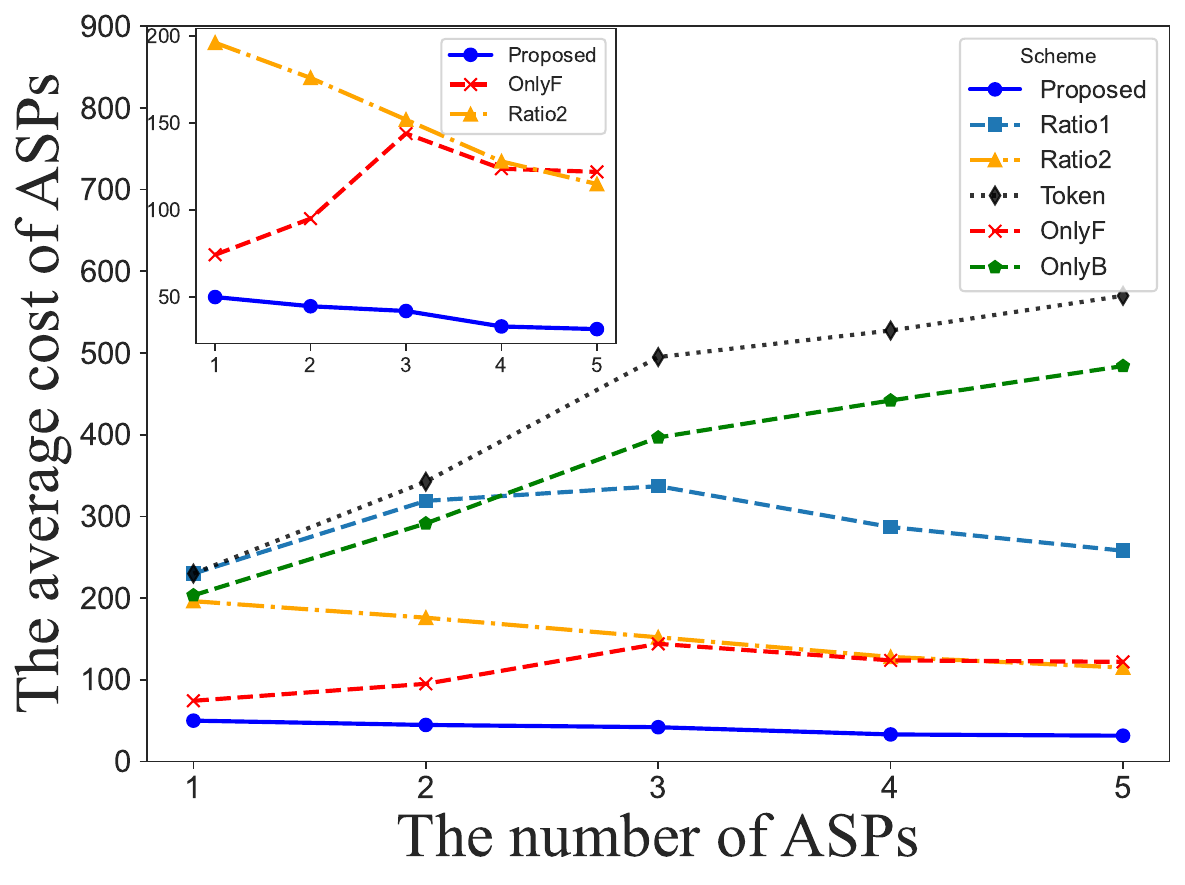}
	}
	\caption{The impact of the number of ASPs on different schemes.}
	\label{compare_ASP}
\end{figure*}

\subsection{Case of Two ASPs and Three MUs}
Based on a set of predefined personalized demands, the QoE-driven incentive mechanism optimally determines rewards for each MU and allocates resources for each ASP. This study highlights the varied AIGC service demands of three distinct MUs, including voice-guided navigation for autonomous driving and academic writing services. As shown in Table~\ref{table1}, MUs prioritize accuracy and concise responses to enhance safety and overall experience in autonomous driving. Conversely, academic writing services focus on the richness and granularity of AIGC.

Table~\ref{table2} presents the results obtained through the QoE-driven incentive mechanism. The two ASPs adopt different resource allocation strategies based on the specific service demands of the MUs.
For example, MU 3 has high demands for both accuracy and the number of output tokens from ASP 1, leading to the highest rewards allocated to ASP 1. In response, ASP 1 allocates more computational and communication resources to meet MU 3's service requests. However, despite the higher reward, the QoE for MU 3 is the lowest.
This is because meeting MU 3's more stringent requirements requires more resources, increasing the total cost and pushing the completion time closer to the maximum tolerable limit. Under our QoE formulation, a lower value directly corresponds to a reduced latency margin, signifying a poorer quality of experience in terms of service timeliness.

The case study illustrates the practical value of the incentive mechanism by effectively addressing the personalized demands of the MUs. This not only alleviates the resource allocation challenges faced by ASPs but also enhances MUs' satisfaction, thereby improving the overall operational efficiency and competitiveness of the service ecosystem.

\subsection{Comparative Analysis}
We conduct a comparative analysis with four incentive design schemes, which can be divided into two categories: (1) joint resource optimization schemes under different reward mechanisms, and (2) different resource allocation schemes under the same reward mechanism. This comparison demonstrates that our proposed scheme effectively reduces service access costs for MUs while minimizing resource consumption of ASPs. The four schemes are summarized as follows:
\begin{itemize}{
	\item {\textbf{\textit{Ratio-based (fixed total reward)}}: Applying the ratio idea from \cite{liao2025reward}, MUs decide the reward of unit QoE proportionally based on a fixed total price $R_{total}$. The total price can be flexibly decided, e.g., $R_{total}=5$ or $R_{total}=1$. ASPs jointly optimize the computational and communication resource allocation based on the rewards and resource constraints to satisfy the personalized demands of MUs.}
    
    	\item {\textbf{\textit{Token-based}}: MUs dynamically adjust rewards based on the number of input and output tokens, similar to the pricing model used by ChatGPT\footnote{https://openai.com/api/pricing/}. ASPs jointly optimize resource allocation based on the rewards and token consumption by MUs.}
        
		\item {\textbf{\textit{OnlyF}}: Inspired by the optimization ideas of \cite{unified,optimal-f1}, each MU determines its rewards based on the expected QoE. ASPs optimize computational resource allocation, while communication resources are equally distributed among MUs.}
      
		\item {\textbf{\textit{OnlyB}}: Similar to \textit{OnlyF}, the reward in \textit{OnlyB} is based on QoE. However, ASPs focus on optimizing communication resource allocation~\cite{optimal_B2,optimal_B1}, while computational resources are equally divided among MUs.}}
\end{itemize}

We evaluate the \textit{proposed} incentive mechanism against the above schemes across resource usage ratio and average cost of MUs (ASPs). The results demonstrate that the \textit{proposed} scheme achieves superior efficiency, scalability, and cost-effectiveness through a QoE-driven incentive mechanism and joint resource optimization.

As shown in Figs.~\ref{utif_MU} and~\ref{utiB_MU}, the \textit{proposed} scheme reduces average resource usage ratio by approximately $75\%$,  $68\%$ and $77\%$ compared to \textit{Ratio1} ($R_{total}=5$), \textit{Ratio2} ($R_{total}=1$) and \textit{Token} schemes. This improvement is reflected in the average cost of the ASPs, as illustrated in Fig.~\ref{ASPcost_MU}, where the \textit{proposed} consistently maintains the lowest cost. Furthermore, the non-cooperative game between MUs leads to a stable utility, with the \textit{proposed}'s average MU cost fluctuating narrowly between $19.09$ and $21.08$ (Fig.~\ref{MUcost_MU}), avoiding the extreme volatility seen in \textit{Ratio1} ($2,468-3,223$) and \textit{Ratio2} ($488-729$). These results stem from the ability of the \textit{proposed} mechanism to dynamically adjust rewards based on QoE, enabling cost control in the competition. In contrast, Ratio-based schemes have high unit reward due to the high $R_{total}$ setting, forcing ASPs to over-provision resources (e.g., \textit{Ratio1}'s computational resource usage ratio is $1.0$ for $20+$ MUs). The \textit{Token} scheme determines the unit reward based solely on the total number of input and output tokens without considering the personalized demands of MUs. Such a model leads to resource wastage, especially when demand for certain resources is lower than expected, and it fails to adapt to dynamic load changes.

For both \textit{OnlyF} and \textit{OnlyB}, the average MU cost and utilization of the optimized resource decrease. \textit{OnlyF} fully utilizes communication resources ($1.0$) while incurring a lower average MU cost and consuming less computational resources than \textit{proposed}. However, this increases the average ASP cost due to over-provisioning communication resources. Similarly, in \textit{OnlyB}, the average MU cost decreases from $126.52$ to $26.03$, but the ASP cost surges to $419.68$. This is because the computational resources are allocated equally without considering actual needs, leading to resource wastage and increased costs. By jointly optimizing both computational and communication resources, the \textit{proposed} avoids these pitfalls and achieves resource allocation on demand.

As shown in Fig.~\ref{compare_ASP}, as the number of ASPs (i.e., parallel tasks) increases, the \textit{proposed} saves approximately $73\%$ in computational and communication overhead, $75\%$ in average cost of MUs, and $84\%$ in average cost of ASPs compared to the other schemes. This is because, when new ASPs are added, MUs, aiming to maximize their benefits, reduce the rewards given to each ASP, resulting in a decrease in both the average cost of MUs and ASPs. This phenomenon reflects the \textit{proposed} scheme's adaptability and on-demand resource allocation capabilities. \textit{Ratio2} clearly outperforms \textit{Ratio1}, primarily due to the impact of the fixed total price $R_{total}$. If $R_{total}$ is set too high, it can lead to overuse of resources; conversely, if set too low, it may fail to meet the service requirements of MUs. This method limits flexibility and adaptability, causing unnecessary resource waste or ineffective service. In addition, the average cost of MUs and ASPs for \textit{OnlyF} and \textit{OnlyB} exhibit variations influenced by the total resources and cost factors, indicating that single resource optimization is insufficient to maximize benefits effectively. For \textit{Token}, the average cost of ASPs continues to increase because the resource usage ratio is always kept at $1.0$. It cannot be dynamically adjusted according to the actual load, resulting in inefficient use of resources.

In summary, existing schemes, including fixed total reward-based approaches (\textit{Ratio1}/\textit{Ratio2}), single-dimensional optimization strategies (\textit{OnlyF}/\textit{OnlyB}), and token-based rewards (\textit{Token}), exhibit limitations in balancing the interests of both ASPs and MUs. In contrast, the \textit{proposed}, leveraging QoE-driven dynamic rewards and joint resource optimization, enables efficient on-demand resource allocation. As a result, it achieves average reductions of approximately $64.9\%$ in computational and communication overhead, $66.5\%$ in the service cost for MUs, and $76.8\%$ in the resource consumption of ASPs.

\section{Conclusion}\label{VII}
In this paper, we addressed the challenge of delivering personalized AIGC services in resource-constrained edge networks. Our key contribution was the design of a novel multi-dimensional QoE metric that effectively captured the personalized demands of MUs by incorporating AIGC accuracy, token count, and service latency. We quantified model accuracy by measuring the performance gap in chain-of-thought (CoT) reasoning and validated its effectiveness using the ScienceQA benchmark.
To cope with limited ASP resources and the competition among MUs, we developed a QoE-driven incentive mechanism for personalized service provisioning. We reformulated the incentive process as an EPEC, proved the existence and uniqueness of the equilibrium, and proposed a dual-perturbation reward optimization algorithm to compute the optimal solution.
Extensive experiments demonstrated that our mechanism achieves efficient, on-demand resource allocation while significantly reducing the overall cost for both ASPs and MUs.

\appendices

\section{Proof of Convergence}

\begin{theorem*}
Suppose the following assumptions~\cite{6021434} hold:
\begin{itemize}
    \item[\textbf{(A1)}] Each utility $U_m^{mu}(\bm{R}_m)$ is strictly concave and continuously differentiable in $\bm{R}_m$,
    \item[\textbf{(A2)}] The perturbation directions are bounded, i.e., $\|\mathbf{d}_m^{(t)}\|_2^2 \leq d_{m}^{+}$, for all $ m \in \mathcal{M}$,
    \item[\textbf{(A3)}] The difference of rewards between two iterations is bounded, i.e., $\|\mathbf{R}_m^{(t)} - \mathbf{R}_m^{(s)}\|_2^2 \leq R_m^+, \ \forall t \neq s, \ t,s \in \{1, \dots, T\}$, where $R_m^+$ is a constant.
\end{itemize}
Then, under different step size settings, the convergence properties of Algorithm~1 are as follows:

\begin{enumerate}
    \item[\textbf{(i)}]
    If the step sizes $\Delta_{t}$ satisfy
    $
    \sum_{t=1}^\infty \Delta_{t} = \infty,  \sum_{t=1}^\infty \Delta_{t}^2 < \infty$, 
    then Algorithm~1 converges to an $\epsilon$-Nash equilibrium (NE), i.e.,
    \[
    \frac{1}{T_{\epsilon}} \sum_{t=1}^{T_{\epsilon}} \left[ \sum_{m=1}^M \left( U_m^{mu}(\mathbf{R}_m^{*}) - U_m^{mu}(\mathbf{R}_m^{(t)}) \right) \right] \leq \epsilon,
    \]
    within 
    $
    T_\epsilon = \mathcal{O}\left( \exp\left( \frac{\ddot{\kappa}}{\epsilon} \right) \right)$ iterations,
    where $\ddot{\kappa}=  (\frac{N}{u}\sum_{m=1}^M(R_{m}^{max})^2 +  MN\frac{u\pi^2}{6})$.

    \item[\textbf{(ii)}]
    If the step size is fixed as $\Delta_{t} = \hat{\Delta}$, then Algorithm~1 approaches a
    $\left[ \frac{\hat{\Delta} d_{1}^{+}}{2}, \frac{\hat{\Delta} d_{2}^{+}}{2}, \ldots, \frac{\hat{\Delta} d_{M}^{+}}{2} \right]$
    neighborhood of a NE at a rate of
    $
    \mathcal{O}\left( \frac{\kappa'}{\epsilon'} \right)$,
    where 
   $\epsilon' =  \epsilon - \frac{\hat{\Delta} MN}{2}, \quad \kappa' = \frac{N \sum_{m=1}^M (R_m^{\max})^2}{\hat{\Delta}}$, valid for all $\epsilon>\frac{\hat{\Delta} M N}{2}$.
\end{enumerate}
\end{theorem*}

\vspace{5pt}
\begin{proof}
We prove convergence by analyzing the total utility gap and the distance to the NE point $\mathbf{R}^* = (\mathbf{R}_1^*, \dots, \mathbf{R}_M^*)$, where $\mathbf{R}_m^*$ uniquely maximizes $U_m^{mu}(\mathbf{R}_m, \mathbf{R}_{-m}^*)$ for each MU $m \in \mathcal{M}$. The proof proceeds in four steps.

\vspace{8pt}

\noindent \textbf{Step 1: Direction Selection via Two-Sided Finite Difference}

For each $R_{nm}^{(t)}$, the direction $\mathbf{d}_{nm}^{(t)}$ is designed to approximate the gradient sign using a two-sided finite difference. Given the utility function $U_m^{mu}(\mathbf{R}_m)$ is continuously differentiable (A1), we evaluate the utility at perturbed rewards $R_{nm}^{(t)} \pm \Delta_{t}$:
\begin{equation}
\begin{aligned}
    U_m^{mu}(R_{nm}^{(t)} + \Delta_{t}, \mathbf{R}_{-n,m}^{(t)}) \approx &\ U_m^{mu}(R_{nm}^{(t)}, \mathbf{R}_{-n,m}^{(t)}) + \Delta_{t} \frac{\partial U_m^{mu}}{\partial R_{nm}^{(t)}} \\&+ \frac{\Delta_t^2}{2} \frac{\partial^2 U_m^{mu}}{\partial (R_{nm}^{(t)})^2} + O(\Delta_t^3),
\end{aligned}
\end{equation}
\begin{equation}
\begin{aligned}
     U_m^{mu}(R_{nm}^{(t)} - \Delta_t, \mathbf{R}_{-n,m}^{(t)}) \approx &\ U_m^{mu}(R_{nm}^{(t)}, \mathbf{R}_{-n,m}^{(t)}) - \Delta_t \frac{\partial U_m^{mu}}{\partial R_{nm}^{(t)}} \\&+ \frac{\Delta_t^2}{2} \frac{\partial^2 U_m^{mu}}{\partial (R_{nm}^{(t)})^2} + O(\Delta_t^3).
\end{aligned}
\end{equation}
Subtracting the above equations gives the utility difference:
\begin{equation}
\begin{aligned}
     U_m^{mu}(R_{nm}^{(t)} + \Delta_t, \mathbf{R}_{-n,m}^{(t)}) - U_m^{mu}(R_{nm}^{(t)} - \Delta_t, \mathbf{R}_{-n,m}^{(t)}) \\ \approx 2\Delta_t \frac{\partial U_m^{mu}}{\partial R_{nm}^{(t)}} + O(\Delta_t^3).
\end{aligned}
\end{equation}
Thus, the gradient estimate becomes:
\begin{equation}
\begin{aligned}
    \hat{g}_{nm}^{(t)} &= \frac{U_m^{mu}(R_{nm}^{(t)} + \Delta_t, \mathbf{R}_{-n,m}^{(t)}) - U_m^{mu}(R_{nm}^{(t)} - \Delta_t, \mathbf{R}_{-n,m}^{(t)})}{2\Delta_t} \\ &\approx \frac{\partial U_m^{mu}}{\partial R_{nm}^{(t)}} + O(\Delta_t^2).
\end{aligned}
\end{equation}
The corresponding direction is then selected as:
\begin{equation}
    \mathbf{d}_{nm}^{(t)} = \text{sign}(\hat{g}_{nm}^{(t)}) =
\begin{cases} 
1 & \text{if } \hat{g}_{nm}^{(t)} > 0, \\
-1 & \text{if } \hat{g}_{nm}^{(t)} < 0, \\
0 & \text{if } \hat{g}_{nm}^{(t)} \approx 0.
\end{cases}
\end{equation}
Since \(\hat{g}_{nm}^{(t)} \approx \frac{\partial U_m^{mu}}{\partial R_{nm}^{(t)}}\), this yields:
\begin{equation}
    d_{nm}^{(t)} \approx \text{sign}\left( \frac{\partial U_m^{mu}}{\partial R_{nm}^{(t)}} \right), \mathbf{d}_{m}^{(t)} \approx \text{sign}(\nabla U_m^{mu}(\mathbf{R}_m^{(t)})),
\end{equation}
with an approximation error of $O(\Delta_t^2)$. This two-sided finite difference approach provides a more accurate gradient sign estimate than single-sided comparisons, enhancing robustness and aligning with zeroth-order optimization techniques.

\vspace{8pt}

\noindent \textbf{Step 2: Monotonic Utility Improvement and Utility Gap}  

Since $\mathbf{d}_{m}^{(t)} \approx \text{sign}(\nabla U_m^{mu}(\mathbf{R}_m^{(t)}))$, we use $\mathbf{d}_{m}^{(t)}$ as the ascent direction for $U_m^{mu}$ in the following analysis. The update:
\begin{equation}
    \mathbf{R}_m^{(t+1)} = \mathbf{R}_m^{(t)} + \Delta_t \mathbf{d}_{m}^{(t)},
\end{equation}
implements a sign-based gradient ascent on $U_m^{mu}$. Since $U_m^{mu}$ is strictly concave (A1), if $\mathbf{d}_{nm}^{(t)} \neq 0$, then $U_m^{mu}(\mathbf{R}_m^{(t+1)}) > U_m^{mu}(\mathbf{R}_m^{(t)})$. If $\mathbf{d}_{nm}^{(t)} = 0$, the utility remains unchanged. Thus, the sequence $\{U_m^{mu}(\mathbf{R}_m^{(t)})\}_{t=1}^\infty$ is non-decreasing for each $m$.
The algorithm aims to minimize the utility gap:
\begin{equation}
    D_t = \sum_{m=1}^M \left( U_m^{mu}(\mathbf{R}_m^*) - U_m^{mu}(\mathbf{R}_m^{(t)}) \right) \geq 0,
\end{equation}
where \(\mathbf{R}_m^*\) is the reward vector maximizing \( U_m^{mu}(\mathbf{R}_m) \). For concave functions (A1), the subgradient inequality holds~\cite{BoydVandenberghe2004}:
\begin{equation}
U_m^{mu}(\mathbf{R}_m^*) \leq U_m^{mu}(\mathbf{R}_m^{(t)}) + \nabla U_m^{mu}(\mathbf{R}_m^{(t)})^\top (\mathbf{R}_m^* - \mathbf{R}_m^{(t)}).
\end{equation}
Summing over all MUs:
\begin{equation} 
D_t \leq \sum_{m=1}^M \nabla U_m^{mu}(\mathbf{R}_m^{(t)})^\top (\mathbf{R}_m^* - \mathbf{R}_m^{(t)}).
\end{equation}
Using \(\mathbf{d}_{m}^{(t)} \approx \text{sign}(\nabla U_m^{mu}(\mathbf{R}_m^{(t)}))\), we approximate:
\begin{equation}
    \sum_{m=1}^M \nabla U_m^{mu}(\mathbf{R}_m^{(t)})^\top (\mathbf{R}_m^* - \mathbf{R}_m^{(t)}) \approx (\mathbf{d}^{(t)})^\top (\mathbf{R}^* - \mathbf{R}^{(t)}),
\end{equation}where $\mathbf{d}^{(t)} = (\mathbf{d}_1^{(t)}, \dots, \mathbf{d}_M^{(t)})$ and $\mathbf{R}^{(t)} = (\mathbf{R}_1^{(t)}, \dots, \mathbf{R}_M^{(t)})$. 
Therefore, the utility gap is upper bounded by:
\begin{equation}
    D_t \leq (\mathbf{d}^{(t)})^\top (\mathbf{R}^* - \mathbf{R}^{(t)}) \implies (\mathbf{d}^{(t)})^\top (\mathbf{R}^{(t)} - \mathbf{R}^*) \leq -D_t.
\end{equation}

\vspace{8pt}

\noindent \textbf{Step 3: Distance to Optimal Point}  

The feasible set is compact, with \(\|\mathbf{R}_m^{(t)} - \mathbf{R}_m^{(s)}\|_2^2 \leq R_m^+\), where \( R_m^+ = N (R_m^{\max})^2 \)(A3). Thus, the total squared distance across all MUs is upper bounded:
\begin{equation}
    \|\mathbf{R}^{(t)} - \mathbf{R}^{(s)}\|_2^2 \leq R^{+}, \ R^{+} = \sum_{m=1}^M R_m^+ = N \sum_{m=1}^M (R_m^{\max})^2.
\end{equation}
Next, consider the squared Euclidean distance between the current reward profile and the optimal $\mathbf{R}^*$:
\begin{equation}
    \begin{aligned}
        &\|\mathbf{R}^{(t+1)} - \mathbf{R}^*\|_2^2 = \|\mathbf{R}^{(t)} + \Delta_t \mathbf{d}^{(t)} - \mathbf{R}^*\|_2^2 \\&= \|\mathbf{R}^{(t)} - \mathbf{R}^*\|_2^2 + 2 \Delta_t (\mathbf{d}^{(t)})^\top  (\mathbf{R}^{(t)} - \mathbf{R}^*) + \Delta_t^2 \|\mathbf{d}^{(t)}\|_2^2.
    \end{aligned}
\end{equation}
Using the inequality $(\mathbf{d}^{(t)})^\top  (\mathbf{R}^{(t)} - \mathbf{R}^*) \leq -D_t$ from \textbf{Step 2}, we obtain:
\begin{equation}
    \|\mathbf{R}^{(t+1)} - \mathbf{R}^*\|_2^2 \leq \|\mathbf{R}^{(t)} - \mathbf{R}^*\|_2^2 - 2 \Delta_t D_t + \Delta_t^2 \|\mathbf{d}^{(t)}\|_2^2.
\end{equation}

\vspace{8pt}

\noindent \textbf{Step 4: Convergence Rate}  

By averaging the inequality in \textbf{Step 3} over iterations $t = 1, \dots, T $:
\begin{equation}
\begin{aligned}
    \frac{1}{T} \sum_{t=1}^T \|\mathbf{R}^{(t+1)} - \mathbf{R}^*\|_2^2 
\leq &\ 
\frac{1}{T} \sum_{t=1}^T \|\mathbf{R}^{(t)} - \mathbf{R}^*\|_2^2 
\\&- \frac{2}{T} \sum_{t=1}^T \Delta_t D_t 
+ \frac{1}{T} \sum_{t=1}^T \Delta_t^2 \|\mathbf{d}^{(t)}\|_2^2.
\end{aligned}
\end{equation}
This can be rewritten recursively as:
\begin{equation}\label{rewritten}
\begin{aligned}
    \frac{1}{T} \sum_{t=1}^T \| \mathbf{R}^{(t+1)} - \mathbf{R}^* \|_2^2 
\leq &\ 
\| \mathbf{R}^{(1)} - \mathbf{R}^* \|_2^2 
- \frac{2}{T} \sum_{t=1}^T \sum_{l=1}^t \Delta_l D_l 
\\&+ \frac{1}{T} \sum_{t=1}^T \sum_{l=1}^t \Delta_l^2 \| \mathbf{d}^{(l)} \|_2^2.
\end{aligned}
\end{equation}
Since $\sum_{l=1}^{t} \frac{1}{l} \geq \frac{1}{T} \sum_{t=1}^{T} \sum_{l=1}^{t} \frac{1}{l}$,
the
second term in the right-hand-side of (\ref{rewritten}) can be expressed as:
\begin{equation}\label{term}
    \begin{aligned}
        \frac{2}{T} \sum_{t=1}^{T}\sum_{l=1}^{t} \Delta_{l}D_{l} &\geq \frac{2}{T}\sum_{t=1}^{T}\left ( \sum_{l=1}^{t} \Delta_{l} \right )\min_{l\in [1,t]} D_{l} \\
&\geq \left (\frac{2}{T}\sum_{t=1}^{T} \sum_{l=1}^{t} \Delta_{l} \right )\frac{1}{T}\sum_{t=1}^{T}\min_{l\in [1,t]} D_{l},
    \end{aligned}
\end{equation}
Substituting (\ref{term}) into (\ref{rewritten}), we derive the following bound:
\begin{equation}
\begin{aligned}
    \frac{1}{T} \sum_{t=1}^T \min_{l \in [1,t]} D_l 
\leq &\ 
\frac{ \| \mathbf{R}^{(1)} - \mathbf{R}^* \|_2^2 
+ \frac{1}{T} \sum_{t=1}^T \sum_{l=1}^t \Delta_l^2 \| \mathbf{d}^{(l)} \|_2^2}
{\frac{2}{T} \sum_{t=1}^T \sum_{l=1}^t \Delta_l }\\&-\frac{\frac{1}{T} \sum_{t=1}^T \| \mathbf{R}^{(t+1)} - \mathbf{R}^* \|_2^2 }{\frac{2}{T} \sum_{t=1}^T \sum_{l=1}^t \Delta_l}.
\end{aligned}
\end{equation}
Since $\frac{1}{T} \sum_{t=1}^T \| \mathbf{R}^{(t+1)} - \mathbf{R}^* \|_2^2 \geq 0$ and $\|\mathbf{R}^{(1)} - \mathbf{R}^*\|_2^2 \leq R^{+}$, it follows that:
\begin{equation}\label{20}
    \frac{1}{T} \sum_{t=1}^T \min_{l \in [1,t]} D_l 
\leq 
\frac{ R^{+} + \frac{1}{T} \sum_{t=1}^T \sum_{l=1}^t \Delta_l^2 \| \mathbf{d}^{(l)} \|_2^2 }
{\frac{2}{T} \sum_{t=1}^T \sum_{l=1}^t \Delta_l }.
\end{equation}
Substituting the step size $\Delta_l = \frac{u}{l}$ into (\ref{20}) , and using the facts that $\|\mathbf{d}^{(t)}\|_2^2 \leq M N$ (A2), $\Delta_l^2 = \frac{u^2}{l^2}$, and the inequality $\sum_{l=1}^t\frac{1}{l^2}<\sum_{l=1}^{\infty}\frac{1}{l^2}=\frac{\pi^2}{6}$, we obtain:
\begin{equation}
    \frac{1}{T} \sum_{t=1}^T \sum_{l=1}^t \Delta_l^2 \|\mathbf{d}^{(l)}\|_2^2 \leq \frac{u^2 M N}{T} \sum_{t=1}^T \sum_{l=1}^t \frac{1}{l^2} \leq M N u^2 \frac{\pi^2}{6}.
\end{equation}
Moreover, leveraging the properties of the harmonic number~\cite{BoydVandenberghe2004}: $\sum_{l=1}^t \frac{1}{l} = \log t + \frac{1}{2t} - \frac{1}{12 t^2} + O(t^{-4}) \approx \log t + \frac{1}{2t}+\dot{\kappa}$, where 
$\dot{\kappa}$ denotes a small residual constant. Based on this approximation, the following result can be derived:
\begin{equation}\label{22}
\begin{aligned}
    \frac{1}{T} \sum_{t=1}^T \min_{l \in [1,t]} D_l 
&\leq 
\frac{ \ddot{\kappa}}
{\frac{2}{T} \sum_{t=1}^T (\log t + \frac{1}{2t}+\dot{\kappa}) }\\& \leq 
\frac{\ddot{\kappa}}{2 \log T + 2\dot{\kappa}} \approx \frac{\ddot{\kappa}}{2 \log T},
\end{aligned}
\end{equation}where $\ddot{\kappa}=  (\frac{N}{u}\sum_{m=1}^M(R_{m}^{max})^2 +  MN\frac{u\pi^2}{6})$.

Since $\sum_{t=1}^T \log t = \Theta(T \log T)$ and hence $\lim\limits_{T \to \infty}\frac{1}{T} \sum_{t=1}^T \log t \to \infty$, which implies that Algorithm 1 converges to a stationary point. Assume that for $T_{\epsilon}\geq T$ and $T_{\epsilon}$ is
a very large number, we try to achieve the accuracy of $\frac{1}{T_\epsilon} \sum_{t=1}^{T_\epsilon} \min_{l \in [1,t]} D_l = \epsilon$. From (\ref{22}) and concavity of logarithm function, we can obtain:
\begin{equation}
    \epsilon \leq \frac{\ddot{\kappa}}{2 \log T_\epsilon} \implies \log T_\epsilon \leq \frac{\ddot{\kappa}}{\epsilon} \implies T_\epsilon = \mathcal{O}\left(\exp\left(\frac{\ddot{\kappa}}{\epsilon}\right)\right).
\end{equation}
Therefore, Algorithm 1 guarantees convergence to an $\epsilon$-NE at a speed
of $\mathcal{O}\left(\exp\left(\frac{\ddot{\kappa}}{\epsilon}\right)\right)$.

\vspace{16pt}
To reduce the convergence complexity of Algorithm~1, one practical strategy is to relax the required solution accuracy. Specifically, the variable $R_{nm}$ can only be selected from a finite discrete set, i.e.,
$
R_{nm} \in \left\{ R_{m}^{min}, \frac{R_{m}^{max}-R_{m}^{min}}{L}+R_{m}^{min}, \frac{2(R_{m}^{max}-R_{m}^{min})}{L}+R_{m}^{min}, \ldots, R_{m}^{\max} \right\}$,
where $L$ is a fixed quantization level.
This discretization can be effectively modeled by using a constant step size $\hat{\Delta}=\frac{R_{m}^{max}-R_{m}^{min}}{L}$ in the update rule, which transforms the action space into a finite set. 

Substituting $\Delta_l = \hat{\Delta}$ into (\ref{20}), and using the bounded direction $\|\mathbf{d}^{(l)}\|_2^2 \leq M N$ (A2), the numerator of (\ref{20}) is bounded as:
\begin{equation}
    \frac{1}{T} \sum_{t=1}^T \sum_{l=1}^t \Delta^{2}_{l} \|\mathbf{d}^{(l)}\|_2^2 
    \leq \hat{\Delta}^2 MN\frac{1}{T}\sum_{t=1}^T t = \frac{\hat{\Delta}^2 M N(T+1)}{2}.
\end{equation}
Similarly, the denominator becomes:
\begin{equation}
    \frac{2}{T} \sum_{t=1}^T \sum_{l=1}^t \Delta_l = 2 \hat{\Delta} \frac{T(T+1)}{2T} = \hat{\Delta}(T+1).
\end{equation}
Therefore, (\ref{20}) yields:
\begin{equation}
\frac{1}{T} \sum_{t=1}^T \min_{l \in [1,t]} D_l 
\leq 
\frac{R^{+}}{\hat{\Delta}(T+1)} + \frac{\hat{\Delta} M N}{2}.
\end{equation}
Assuming that $T_{\epsilon}\geq T$ and $T_{\epsilon}$ is
a very large number, we try to achieve the accuracy of $\frac{1}{T_\epsilon} \sum_{t=1}^{T_\epsilon} \min_{l \in [1,t]} D_l = \epsilon$. 
Then, we can derive the following inequality:
\begin{equation}
   \epsilon - \frac{\hat{\Delta} M N}{2} \leq \frac{R^{+}}{\hat{\Delta}(T_\epsilon+1)}.
\end{equation}
Let $\epsilon' = \epsilon - \frac{\hat{\Delta} M N}{2}$, valid when $\epsilon>\frac{\hat{\Delta} M N}{2}$, we obtain:
\begin{equation}
\epsilon' \leq \frac{R^{+}}{\hat{\Delta}(T_\epsilon+1)} \Rightarrow T_\epsilon \leq \frac{R^{+}}{\hat{\Delta} \epsilon'} - 1 \implies T_\epsilon = \mathcal{O}\left(\frac{{\kappa}'}{\epsilon'}\right),
\end{equation} where $\kappa'= \frac{ N \sum_{m=1}^M (R_m^{\max})^2}{\hat{\Delta}}$, for all $\epsilon>\frac{\hat{\Delta} M N}{2}$.

 In this case, if the initial point
is not appropriately chosen, Algorithm 1 will be unstable for the first few iterations, i.e., fluctuating between different
neighborhood of NEs. However, when the number of iterations is large, the Algorithm~1 approaches a $[\frac{\hat{\Delta}d_{1}^{+}}{2}, \frac{\hat{\Delta}d_{2}^{+}}{2}, \ldots, \frac{\hat{\Delta}d_{M}^{+}}{2}]$
neighborhood of a NE at a rate of $\mathcal{O}(\frac{{\kappa}'}{\epsilon'})$.

\end{proof}
\bibliographystyle{IEEEtran}

\bibliography{AIGC}

\begin{IEEEbiography}[{\includegraphics[width=1in,height=1.25in,clip]{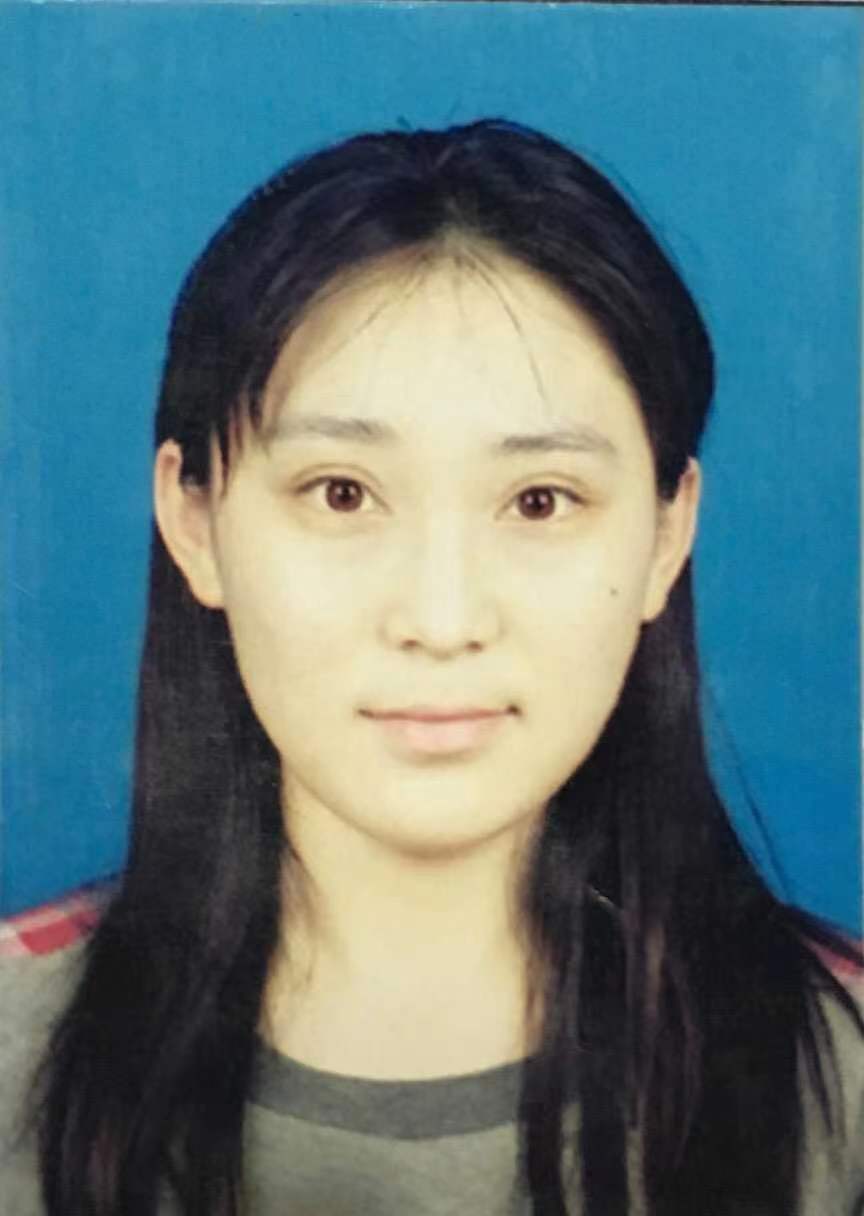}}]{Hongjia Wu} received the Ph.D. degree in computer science from the National University of Defense Technology, China, in 2023. She was a visiting student at the Singapore University of Technology and Design. She is currently a Post-doctoral Fellow at the Education University of Hong Kong in China. Her research interests mainly focus on computational offloading, game theory, dispersed computing and internet of things. 
 \end{IEEEbiography}
 
		\begin{IEEEbiography}[{\includegraphics[width=1in,height=1.25in,clip]{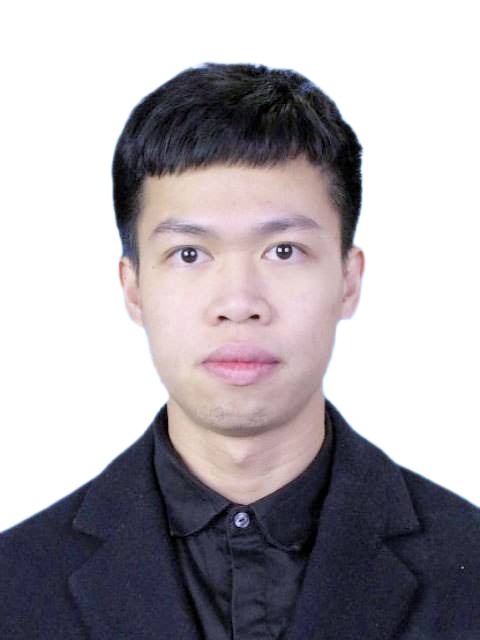}}]{Minrui Xu} (Graduated Student Member, IEEE)
  received the B.S. degree from Sun Yat-Sen University, Guangzhou, China, in 2021. He is currently working toward the Ph.D. degree in the School of Computer Science and Engineering, Nanyang Technological University, Singapore. His research interests mainly focus on Metaverse, deep reinforcement learning, and mechanism design.
 \end{IEEEbiography}

	\begin{IEEEbiography}[{\includegraphics[width=1in,height=1.25in,clip]{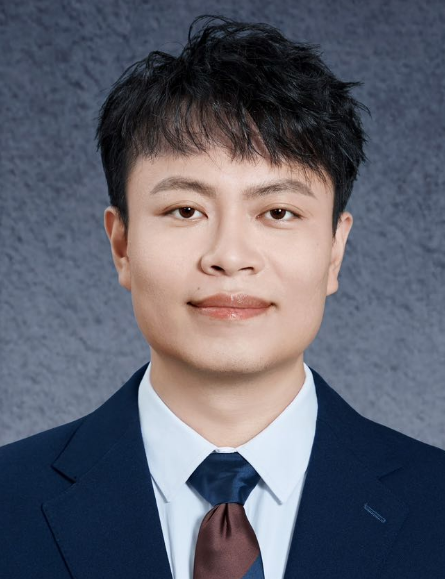}}]{Zehui Xiong} is currently a Full Professor with the School of Electronics, Electrical Engineering and Computer Science, Queen's University Belfast, United Kingdom. Prior to that, he was with Singapore University of Technology and Design, and Alibaba-NTU Singapore Joint Research Institute. He received his Ph.D. degree from Nanyang Technological University and was a visiting scholar with Princeton University and University of Waterloo. Recognized as a Clarivate Highly Cited Researcher, he has published over 250 peer-reviewed research papers in leading journals, with numerous Best Paper Awards from international flagship conferences. His research interests include 5G/B5G/6G and cellular systems, autonomous vehicles and intelligent transportation, cooperative communication and green communication, land transportation, vehicular communications and connected vehicles, wireless networks.
 \end{IEEEbiography}

	\begin{IEEEbiography}[{\includegraphics[width=1in,height=1.25in,clip]{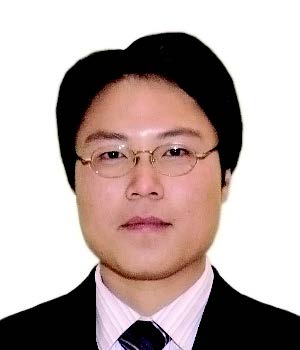}}]{Lin Gao} is a Professor
in the School of Electronic and Information Engineering, Harbin Institute of Technology (HIT)
Shenzhen Graduate School. He received the
M.S. and Ph.D. degrees in Electronic Engineering from Shanghai Jiao Tong University (China)
in 2006 and 2010, respectively. He worked as
a Postdoc Research Fellow at the Chinese University of Hong Kong from 2010 to 2015. His
research interests are in the area of network
economics and games, with applications in wireless communications, networks, and internet of things.
 \end{IEEEbiography}

\begin{IEEEbiography}[{\includegraphics[width=1in,height=1.25in,clip,keepaspectratio]{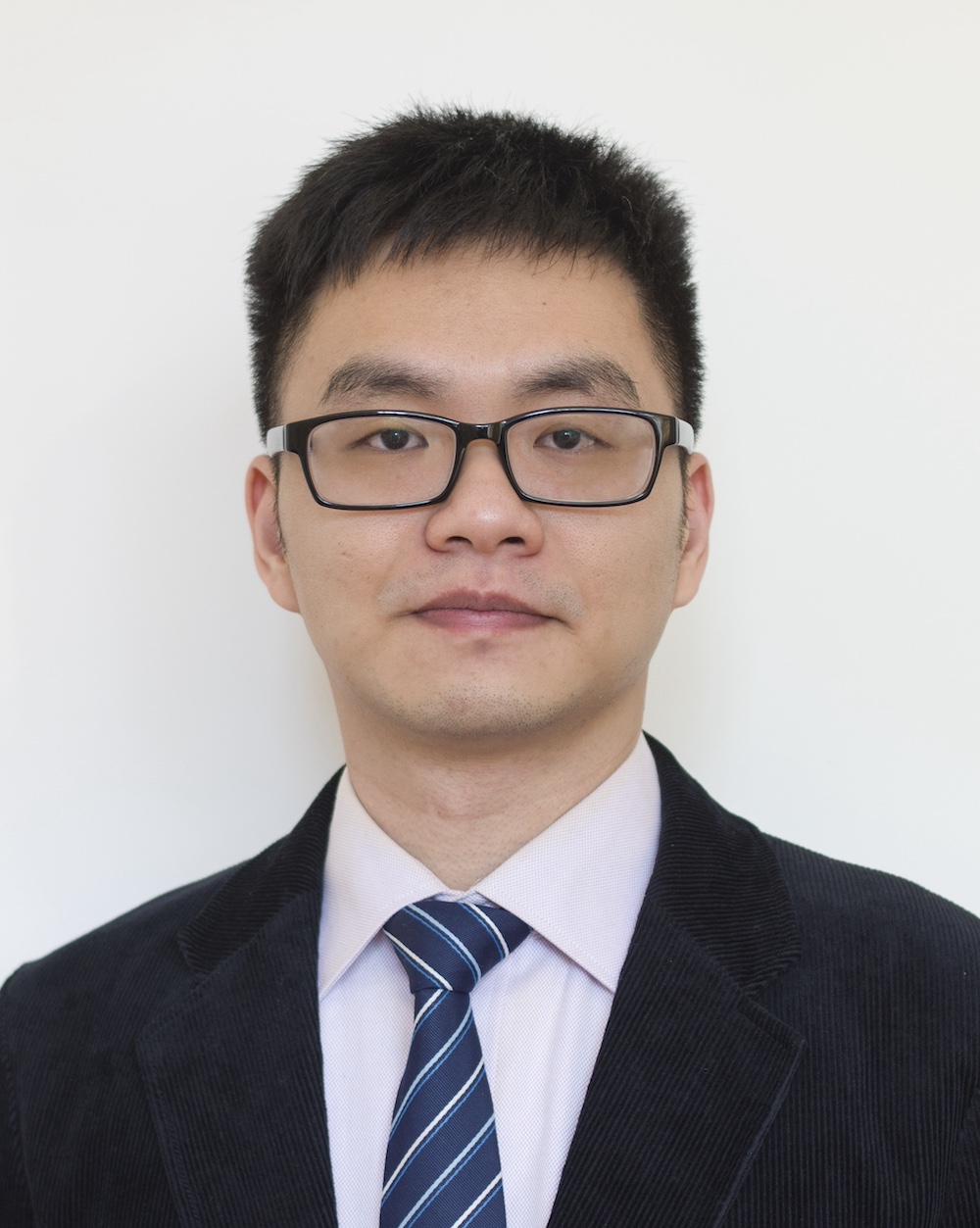}}]{Haoyuan Pan}(Member, IEEE) received the B.E. and Ph.D. degrees in Information Engineering from The Chinese University of Hong Kong (CUHK), Hong Kong, in 2014 and 2018, respectively. 

He was a Post-Doctoral Fellow with the Department of Information Engineering, CUHK, from 2018 to 2020. He is currently an assistant professor with the College of Computer Science and Software Engineering, Shenzhen University, Shenzhen, China. His research interests include wireless communications and networks, Internet of Things (IoT), semantic communications, and age of information (AoI).
\end{IEEEbiography}

\begin{IEEEbiography}[{\includegraphics[width=1in,height=1.25in,clip,keepaspectratio]{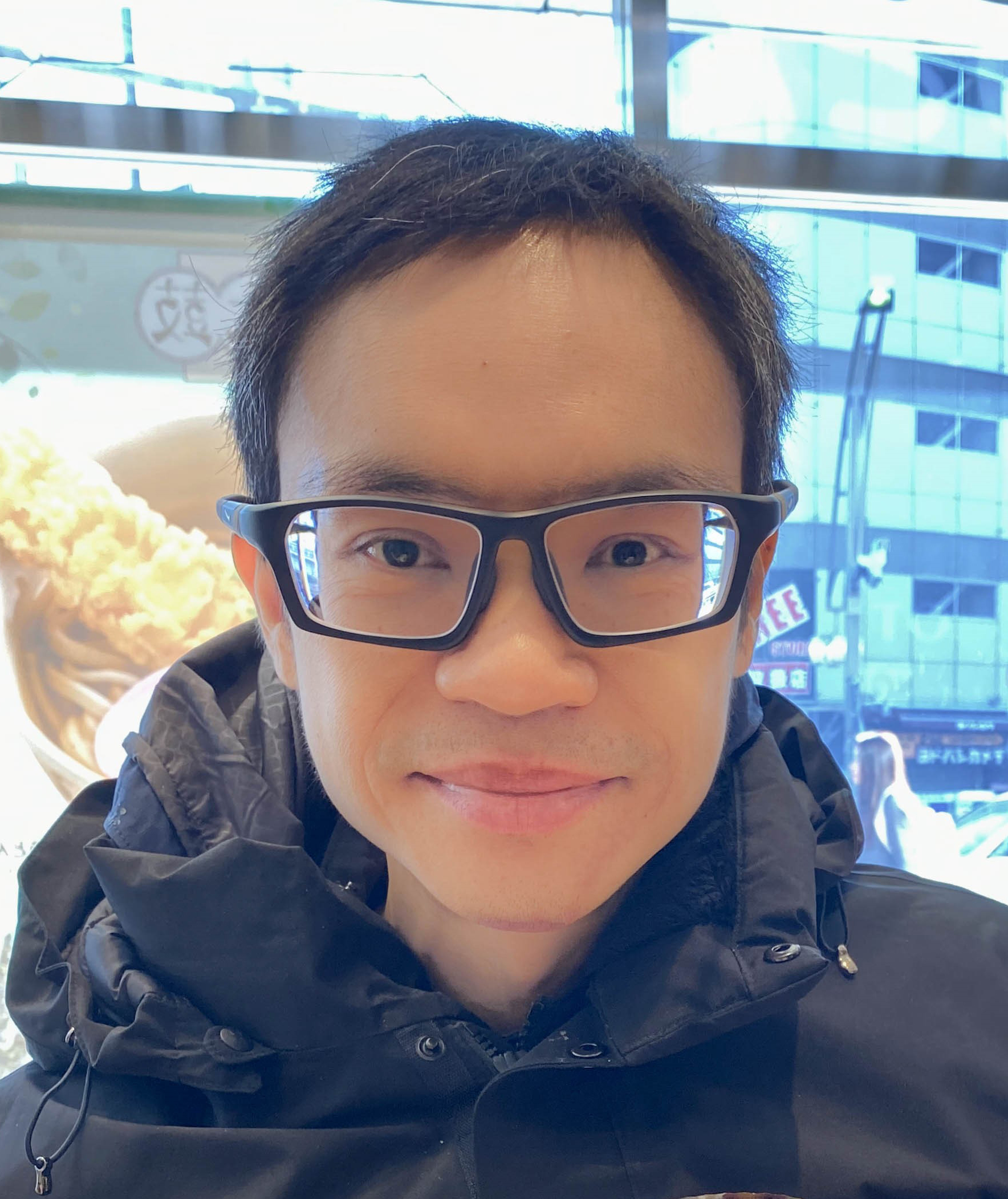}}]{Dusit Niyato}(Fellow, IEEE) is a professor in the School of Computer Science and Engineering, at Nanyang Technological University, Singapore. He received B.Eng. from King Mongkuts Institute of Technology Ladkrabang (KMITL), Thailand in 1999 and Ph.D. in Electrical and Computer Engineering from the University of Manitoba, Canada in 2008. His research interests are in the areas of sustainability, edge intelligence, decentralized machine learning, and incentive mechanism design.
\end{IEEEbiography}

\begin{IEEEbiography}[{\includegraphics[width=1in,height=1.25in,clip,keepaspectratio]{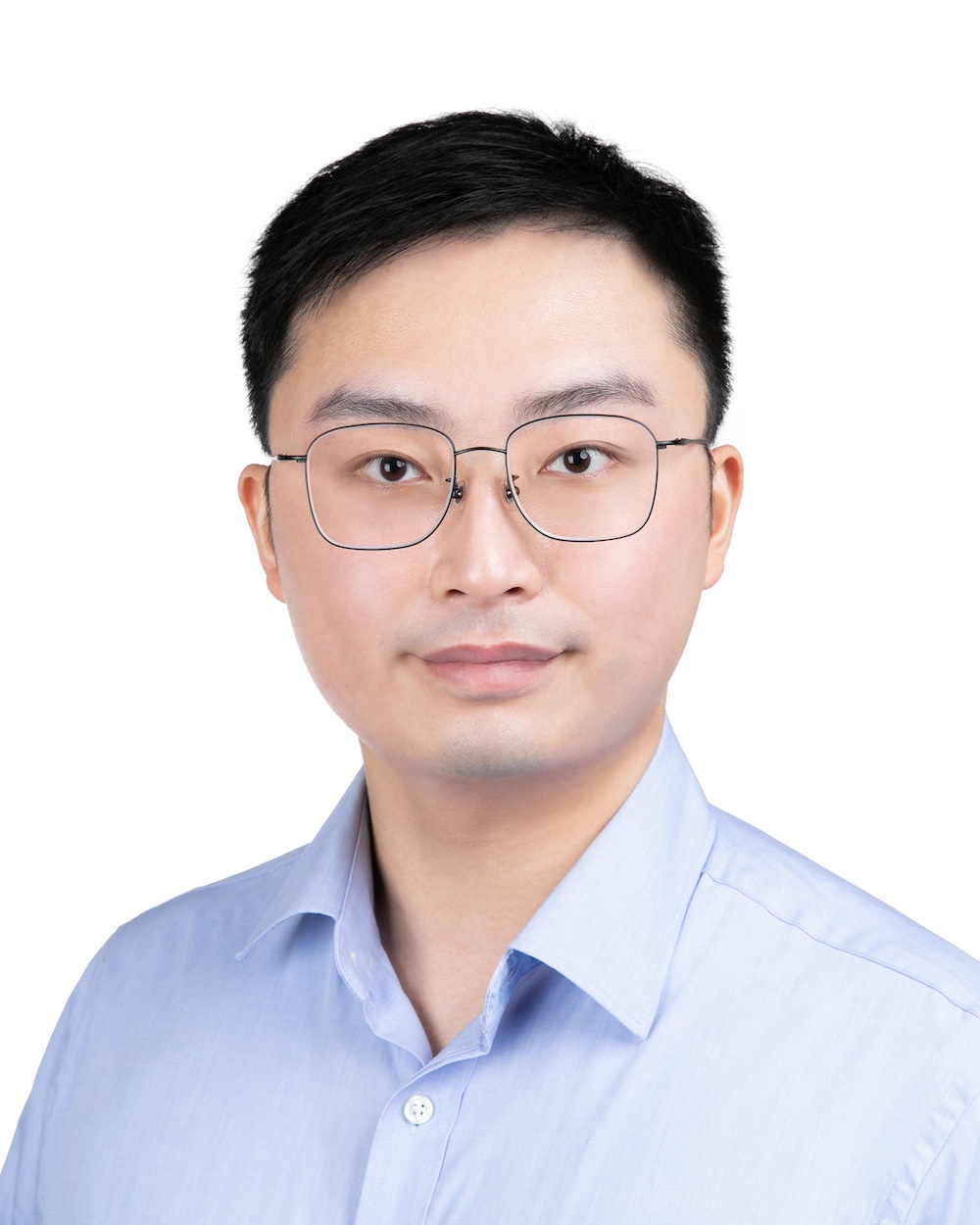}}]{Tse-Tin Chan}(Member, IEEE) received the B.Eng. (First Class Hons.) and Ph.D. degrees in Information Engineering from The Chinese University of Hong Kong (CUHK), Hong Kong SAR, China, in 2014 and 2020, respectively. 

He is currently an Assistant Professor with the Department of Mathematics and Information Technology, The Education University of Hong Kong (EdUHK), Hong Kong SAR, China. From 2020 to 2022, he was an Assistant Professor with the Department of Computer Science, The Hang Seng University of Hong Kong (HSUHK), Hong Kong SAR, China. His research interests include wireless communications and networking, Internet of Things (IoT), age of information (AoI), and artificial intelligence (AI)-native wireless communications.
\end{IEEEbiography}

\end{document}